\documentclass[12pt, twoside]{article}
\usepackage{amsmath}
\usepackage{graphicx}

\newcommand{\blind}{1}

\usepackage{amsthm}
\usepackage{amssymb}
\usepackage{amstext}
\usepackage{float}

\usepackage{geometry}
\usepackage{epsfig,epsf,psfrag,epstopdf}
\usepackage{natbib}
\usepackage{multirow}

\usepackage{url}
\usepackage{setspace}
\usepackage[utf8]{inputenc} 
\usepackage[T1]{fontenc} 
\usepackage{hyperref} 
\hypersetup{ colorlinks, citecolor=blue, filecolor=black, linkcolor=blue, urlcolor=blue } 
\usepackage{bm}

\usepackage{enumitem}
\newcounter{cntr}

\newtheorem{Theorem}{Theorem}

\newtheorem{Proposition}{Proposition}

\newcommand{\reals}{\mathbb R}

\newcommand{\E}[2][]{E_{#1}\left(#2\right)}

\newcommand*{\cT}{\mathcal{T}}
\newcommand{\T}{\cT}

\newcommand{\ind}[1]{\mathbf 1\left(#1\right)}
\newcommand{\ep}{E_{\oplus}}
\newcommand{\no}{\noindent}
\newcommand{\bc}{\begin{center}}
	\newcommand{\ec}{\end{center}}

\begin{document}
	
	\bibliographystyle{ims}
	
	\def\spacingset#1{\renewcommand{\baselinestretch}%
		{#1}\small\normalsize} \spacingset{1}

	\if1\blind
	{
		\title{\bf Concurrent Object Regression}
		\author{Satarupa Bhattacharjee and Hans-Georg M\"uller \footnote{Research supported by NSF grant DMS-1712864.}\\
			Department of Statistics, University of California, Davis\\
			(For the Alzheimer's Disease Neuroimaging Initiative\footnote{Data used in preparation of this article were obtained from the Alzheimer's Disease Neuroimaging Initiative (ADNI) database (adni.loni.usc.edu). As such, the investigators within the ADNI contributed to the design and implementation of ADNI and/or provided data but did not participate in analysis or writing of this report. A complete listing of ADNI investigators can be found at: \url{http://adni.loni.usc.edu/wp-content/uploads/how_to_apply/ADNI_Acknowledgement_List.pdf}})}
		\maketitle
	} \fi
	
	\if0\blind
	{
		\bigskip
		\bigskip
		\bigskip
		\begin{center}
			{\LARGE\bf Concurrent Object Regression}
		\end{center}
		\medskip
	} \fi
	
\begin{abstract}
Modern-day problems in statistics often face the challenge of exploring and analyzing complex non-Euclidean object data that do not conform to vector space structures or operations. Examples of such data objects include covariance matrices, graph Laplacians of networks, and univariate probability distribution functions. In the current contribution a new concurrent regression model is proposed to characterize the time-varying relation between an object in a general metric space (as a response) and a vector in $\reals^p$ (as a predictor), where concepts from Fr\'echet regression is employed. Concurrent regression has been a well-developed area of research for Euclidean predictors and responses, with many important applications for longitudinal studies and  functional data. However, there is no such model available so far for general object data as responses. We develop generalized versions of both global least squares regression and locally weighted least squares smoothing in the context of concurrent regression for responses which are situated in general metric spaces and propose estimators that can accommodate sparse and/or irregular designs. Consistency results are demonstrated for sample estimates of appropriate population targets along with the corresponding rates of convergence. The proposed models are illustrated with human mortality data and resting state functional Magnetic Resonance Imaging data (fMRI) as responses.
		\end{abstract}

\noindent%
\vfill

\newpage
\spacingset{1.5} 
		\section{INTRODUCTION}\label{INTRO}
Concurrent regression models are an important tool to explore the time-dynamic nature of the dependence between two variables. They are often used in regression problems, where the effect of the covariates on the response variable is affected by a  third variable, such as time or age. Specifically, the response at a particular time point is modeled as a function of the value of the covariate only at that specific time point. Concurrent regression models, also known as varying coefficient models, are natural  extensions of (generalized) linear models \citep{hast:tibs:93,clev:17}. Owing to their  interpretability and wide applicability in areas such as economics, finance, politics, epidemiology and the life sciences, there exists  a rich literature on these models that covers  a large range from simple linear models with  scalar responses to more complex longitudinal and functional data \citep{sent:nguy:11,fan:zhang:08,rams:05,sent:mllr:10,horv:koko:12,wang:16}, including regression problems where both responses and covariate(s) are of functional type.

However, as we enter the era of big data, more complex, often non-Euclidean, data are increasingly observed and this motivates the development of statistical models that are suitable for such complex data. In this paper, we introduce Concurrent Object Regression (CORE) models for very general  settings where one is interested in  the time-varying regression relation between a response  that takes values in a general metric object space without any linear structure and real-valued covariate(s). We note that no such models exist at this time and this is the first concurrent model for object data. 

For the special case  where the observations consist of a paired sample of square integrable random functions $(X(t),Y(t))$ that take values in $\reals$, the  linear functional concurrent model is well known \citep{rams:07} and can be written as 
\begin{equation} \label{eq:1.1}
\E{Y(t)|X(t)} = \mu_Y(t) + \beta(t)  \left(X(t) - \mu_X(t)\right),
\end{equation}	
where $\mu_Y(\cdot)$ and $\mu_X(\cdot)$ are respectively the mean functions of $X(\cdot)$ and $Y(\cdot)$ and $\beta(\cdot)$ is the smooth coefficient function. This can be thought of as a series of linear regressions for each time point that are connected and restricted by the assumed smoothness of the coefficient function $\beta$.

Several methods have been proposed to estimate the model components $\mu_X, \mu_Y$  and $\beta$, which are functional in nature,  including local polynomial kernel smoothing regression \citep{fan:99, fan:00, hoov:98, wu:chiang:00}, smoothing splines \citep{eubank:04, chiang:01} and function approximation of $\beta(\cdot)$ through basis expansion \citep{huang:02}. These methods were also adapted for spatial imaging \citep{zhang:11}, ridge regression \citep{manr:18} and other areas. 
Since the linear approach may not capture the true and possibly complex nature of  the relationship between $Y$ and $X$, the response and the covariate, a more general nonparametric model may be preferable,
\begin{equation}\label{eq:1.2}
\E{Y(t) | X(t)} = m(t,X(t)),
\end{equation}	
where the regression function $m$ is assumed to satisfy some basic smoothness properties.

Unlike a linear regression model, the parametric varying coefficient model in~\eqref{eq:1.1} or the nonparametric concurrent model in~\eqref{eq:1.2} involve the nested structure of the predictor space $(T,X(T))$ and allow the regression function (the coefficient functions in the parametric model) to vary systematically and smoothly in more than one direction. We aim to capture the nested predictor space structure and develop a concurrent regression model when the responses are random objects lying in a general metric space. To the best of our knowledge, such a model has not been studied before, even though for its Euclidean analogue various methods have been discussed over the years.

Estimation and inference in the nonparametric functional concurrent regression literature include methodologies such as spline smoothing \citep{maity:17}, Gaussian process regression \citep{shi:07, wang:shi:14}, and local kernel smoothing techniques \citep{verz:12} among others, with various subsequent developments \citep{wang:08, zhu:10}. Regression methods have also been considered more recently for manifold-valued responses in curved spaces \citep{zhu:09, cornea:17, yuan:13, davi:05}, owing to the growing realization that data from many disciplines have manifold structures,  including data generated in brain imaging, medical and molecular imaging, computational biology and computer vision.

The major objective of this paper is to overcome the limitation of Euclidean responses in the previous concurrent regression approaches, where it is always assumed that $Y(t)\in \reals$ or $Y(t) \in \reals^p$. The challenge that one faces in  extending concurrent regression beyond Euclidean responses  is that existing methodology relies in a fundamental way on the vector space structure of the responses,  which is no longer available, not even locally, when responses are situated in general separable metric spaces that cover large classes of possible response types. Technological advances have made it possible to record and efficiently store time courses of image, network, sensor or other complex data. Such ``object-oriented data'' \citep{marr:alon:14}  or ``random objects'' \citep{mull:16} can be viewed as random variables taking  values in a separable metric space that is devoid of a vector space structure and where only pairwise distances between the observed data are available. Such random object data, including distributional data in Wasserstein space~\cite{kato:21, chen:21}, covariance matrix objects~\citep{pete:mull:16}, data on the surface of the sphere~\citep{di:14}, and phylogenetic trees~\citep{bill:01}, have drawn the attention of the statisticians in recent times. 

As a motivating example for the proposed concurrent object regression (CORE), we consider fMRI brain-image scans for Alzheimer's patients over varying ages. It is important to note that, the space $\mathcal{C}$ of the functional connectivity network of fMRI signals, represented as correlation matrices between the different nodes of the brain is not linear and there is no concept of direction.

However, the connectivity correlation matrices can be perceived as random objects in a metric space, endowed with a suitable metric. For example, one might be interested to see if certain measures indexing the advancement of the disease, such as the total cognitive score, are associated with the connectivity matrices. It is known that a higher total cognitive score may be linked with a more serious cognitive deficit and a higher age.  At the same time, the functional connectivity itself is expected to deteriorate with increasing age as the severity of the condition intensifies over time. Of interest is then to ascertain the dependence of the functional connectivity correlation matrices of the Alzheimer's subjects on time (age) and some index of the overall health for the subjects, that also varies over time.

The space of positive semi-definite matrices is a Riemannian manifold which can be flattened locally and analyzed using linear results, however the Riemannian structure of the space depends heavily on the metric. Our approach of treating it as a metric space is more general, in the sense that it works for many metrics in the space such as the Frobenius metric, the log-Euclidean metric \citep{arsi:07}, the Procrustes metric \citep{pigo:14, zhou:16}, the power metric \citep{dryd:09, dryd:10}, the affine-invariant Riemannian metric \citep{penn:06, moak:05}, the Cholesky metric \citep{lin:19} among others. As such we do not have to evoke the Riemannian geometry of the space. However, a possible challenge inherent in Fr\'echet regression to ascertain the existence and uniqueness of the Fr\'echet means may be encountered. 
Other examples of such general metric space objects include time-varying age-at-death densities resulting from demographic data, where the interest is in quantifying the dynamic regression relationship between the densities and time-dependent some economic index such as GDP per capita,
or time-varying network data, for example  internet traffic networks where one has concurrent covariates.

The natural notion of a mean for  random elements of a metric space is the  Fr\'echet mean \citep{frechet:48}. It is a direct generalization of the standard mean, and is defined as the element of the metric space for which the expected squared distance to all other elements is minimized. It can encompass different types of means commonly used, such as the expectation, the median, or the geometric mean, and extends to non-Euclidean spaces, thus allowing for profound applications of probability theory and statistics exploiting the geometry in such spaces \citep{scho:20,sayan:14, yang:20, zhan:21}.
\cite{pete:mull:19} extended the concept of Fr\'echet mean to the notion of a conditional Fr\'echet mean, implemented as Fr\'echet regression, where one has samples of data $(X_i,Y_i)$, with the $Y_i$ being random objects and the $X_i$ are Euclidean predictors. This is an extension of ordinary regression to metric space valued responses.

Even though Fr\'echet regression  \citep{pete:mull:19} can incorporate a random time variable as one of the Euclidean covariates, the concurrent regression relationship between paired stochastic processes of real covariates and an object response as a function of time has not been explored yet. This is an important problem of its own accord and highly relevant in various data applications such as brain imaging for which we provide an example in Section~\ref{DATA:fMRI}. It is of interest to observe that concurrent object regression is not the same as Fr\'echet regression, just as varying coefficient models in~\eqref{eq:1.1} and~\eqref{eq:1.2} are different from linear regression models when the response is Euclidean.

In Section~\ref{NPM}, we introduce the concurrent object regression (CORE) model for time-varying object responses and time-varying real covariate(s). We separately discuss two situations -- one where we assume a ``linear'' dependence of the predictor and response at any given time point and a second scenario in which we assume a nonparametric model in Sections~\ref{NPM} and \ref{PGM} respectively. Our motivating application examples deal with samples of probability distributions, data lying on unit sphere in $\reals^3,$ and correlation matrices, which are illustrated with simulations and real data from neuroimaging and demography, with details in Sections~\ref{SIM} and \ref{DATA}, respectively. We conclude with a brief discussion about our methods in Section~\ref{CONCL} . 

\section{Data and model}
\label{model}
\no Throughout, we consider a totally bounded, hence separable, metric space $(\Omega,d)$, where the response is situated. This is coupled with a $p$-dimensional real valued stochastic process $X(\cdot)$ as a predictor. The  $\Omega$-valued random object  response $Y$ depends on both $X$ and a ``time''-variable $t \in \T$, where $\T$ is a closed and bounded interval on the real line. In other words, $\left(X(t),Y(t)\right) : t \in \T$ are two stochastic processes that, for each given $t$, take values $\reals^p$ and $\Omega$ respectively.

A random time  $T$ is selected from some distribution  $f_T$  on $\T$, at which $X$ is observed.
Note that $X(T)$ is itself a random variable and has a probability distribution on $\reals^p$. The  joint distribution of $(X(T),T)$ is well defined in case $X(T)$ and $T$ are independently distributed. For the sake of generality, we consider the joint distribution of  $(X(T),T)$ and, with a slight abuse of notation, denote the joint distribution by $F_{(X,T)}$, which is a probability distribution on $\reals^p \times \reals$. We further assume that  $Y \sim F_Y$ where $F_Y$ is a distribution on $(\Omega,d)$. The conditional distributions of $Y(T)|(X(T),T)$ and $(X(T),T)|Y(T)$ are denoted by $F_{Y|(X,T)}$ and $F_{(X,T)|Y}$ respectively, assuming they exist.
We define the concurrent object regression (CORE) model as follows
\begin{align}\label{eq:CORE}
m_\oplus(x,t) &:= \ep\left({Y(t)| X(T)=x, T = t}\right) := \underset{\omega \in \Omega}{\text{argmin }}M_\oplus(\omega,x,t),\nonumber \\ 
&M_\oplus(\omega,x,t) = \E{d^2(Y(t),\omega )| X(T)= x, T=t},
\end{align}	
and refer to the objective function $M_\oplus(\cdot,x,t)$ in \eqref{eq:CORE}  as the  conditional Fr\'echet function.

In many scenarios one does  not fully observe the trajectories of responses $Y(t)$ and  covariates $X(t)$.  We  consider a general situation, where each subject is measured at random time points, possibly according to a sparse design, with  observed data  of the form $(T_{il}, X_i(T_{il}),Y_i(T_{il}));\hspace*{.1cm} l=1,\dots, n_i;\, i =1,\dots,n, $ i.e., for  the $i^{\text{th}}$ subject one has observations of the response $Y(\cdot)$ and predictor  $X(\cdot)$ at time points $T_{il}$ that may vary from  subject to subject. We denote the observed data by  $(T_{il}, X_{il},Y_{il});\, l=1,\dots, n_i; \, i =1,\dots,n$.
The number of observations $n_i$ made for the $i^{\text{th}}$ subject is a r.v. with $n_i \overset{i.i.d.}{\sim}N,$
where $N > 0$ is a positive discrete random variable, with $\E{N}<\infty$ and $P(N>1) >0.$
The observation times and measurements are assumed to be independent of the number of measurements, i.e., for any subset $J_i \subseteq \{1,\dots,n_i\}$ and for all $i = 1,\dots, n,$ $\left(\{T_{il} : l \in J_i\},\{X_{il} : l \in J_i\},\{Y_{il} : l \in J_i\}\right)$ is independent of $n_i.$

\section{Nonparametric concurrent object regression}
\label{NPM}
In this section, we develop a nonparametric estimation strategy for the target CORE model~\eqref{eq:CORE}, assuming that the dependence of the response $Y(T)$ on the predictors $X(T)$ and $T,$ for any randomly chosen $T \in \T$ are local, in both directions. For ease of presentation, we provide details for the case of a scalar predictor. For the remainder of this section we will assume that $X(t) \in \reals^p$, where $p = 1$ for all $t \in \T$ , that is the dimension of the predictor space $(T,X(T))$, for any random time point $T$ is $p+1 = 2.$ This allows for simpler notation and implementation. At the cost of much more involved notation, the theory can be extended to cover cases where $p > 1.$

We aim to express the CORE function $m_\oplus(x,t)$ in \eqref{eq:CORE} as a weighted Fr\'echet mean, where the weight function varies with the values $(x,t)$ of the predictors. 
The intuition behind these approaches derives from the special case of Euclidean responses.

As an illustrating motivation, let us first consider here the special case of time-varying Euclidean responses. The space is equipped with the metric $d(a,b) = d_E(a,b) = |a-b|$ for all $a,b \in \reals.$ The minimizer of $M_\oplus$ in \eqref{eq:CORE} exists, is unique and coincides with the conditional expectation, and we write  
\begin{align}\label{eq:SPCL:MODEL}
m_\oplus(x,t) = \ep\left({Y(t)| X(T)=x, T = t}\right) &= \E{Y(t)| X(T)=x, T = t} := m(x,t).
\end{align}
%

Local kernel-based nonparametric regression approaches to estimate a smooth regression function for Euclidean responses have been well investigated due to their versatility and flexibility. If we assume a nonparametric relationship of the response $Y$ with the predictors $T$ and $X(T)$, the local linear estimate of the function $m$  in \eqref{eq:SPCL:MODEL} at any given point $(x,t)$
is given by $\hat{m}(x,t) := \hat{\beta}_0.\text{ Here}$
\begin{align}\label{eq:SPCL:EST}
(\hat{\beta}_0, \hat{\beta}_1, \hat{\beta}_2) = \underset{\beta_0, \beta_1,\beta_2}{\text{argmin  }} \frac{1}{n} \displaystyle \sum_{i=1}^{n} \Bigl( \frac{1}{n_i}\displaystyle \sum_{j=1}^{n_i} & \left(Y_{il}-\beta_0 - \beta_1(X_{il}-x) - \beta_2(T_{il}-t)\right)^2 \Bigr) \nonumber \\ 
& \hspace*{1cm}\times K_{h_1,h_2}(X_{il}-x, T_{il}-t).
\end{align}
$K$ is a bivariate kernel function, which corresponds to a bivariate density function, and $h_1, h_2$ are the bandwidth parameters such that $K_{h_1,h_2}(x_1,x_2)= (h_1h_2)^{-1} K(x_1/h_1, x_2/h_2)$.   We can view the above estimator in~\eqref{eq:SPCL:EST} as an M-estimator of the alternative population target
\begin{align}\label{eq:SPCL:EST2}
(\beta_0^*, \beta_1^*,  \beta_2^*) = \underset{\beta_0, \beta_1,\beta_2}{\text{argmin }} \displaystyle \int & \Bigl[K_{h_1,h_2}(z-x, s-t) \nonumber \\
\times \Bigl(\int ydF_{Y|X,T}&(y,z,s) - \beta_0 - \beta_1(z-x) - \beta_2 (s-t)\Bigr)^2 \Bigr]dF_{(X,T)}(z,s).
\end{align}

Defining
\begin{align}\label{eq:SPCL:EST3}
\mu_{jk} &:= \E{K_{h_1,h_2}(X-x, T-t)(X-x)^j(T-t)^k},\\
r_{jk} &:= \E{K_{h_1,h_2}(X-x, T-t)(X-x)^j(T-t)^kY},\quad
\Sigma=
\begin{bmatrix}
\mu_{00} & \mu_{10} & \mu_{01} \nonumber\\
\mu_{10} & \mu_{20} & \mu_{11} \nonumber \\
\mu_{01} & \mu_{11} & \mu_{02}
\end{bmatrix}, \nonumber
\end{align}
the solution of the minimization problem in \eqref{eq:SPCL:EST2} is 
\begin{align}\label{eq:SPCL:EST4}
\tilde{l}(x,t) = \beta_0^* = \begin{bmatrix} 1, & 0, &  0 \end{bmatrix}\Sigma^{-1} \begin{bmatrix} r_{00}, & r_{10}, & r_{01} \end{bmatrix} = \E{s^L(X,x,T,t,h_1,h_2)Y},
\end{align}
with weight function $s^L$ given by 
\begin{align}\label{eq:SPCL:EST5a}
s^L(X,x,T,t,h_1,h_2) &= K_{h_1,h_2}(X-x, T-t) \left[ \nu_1 + \nu_2 (X-x) +\nu_3 (T-t) \right],\\
[\nu_1, \nu_2, \nu_3] & = \frac{1}{ \sigma^2_0}\left[\mu_{20}\mu_{02}- \mu_{11}^2, \quad \mu_{01}\mu_{11} -  \mu_{02}\mu_{10}, \quad \mu_{10}\mu_{11} - \mu_{20}\mu_{01}\right],\nonumber\\
\sigma^2_0 = |\Sigma| &=  \left(\mu_{00} \mu_{20}\mu_{02} - \mu_{00}\mu_{11}^2 - \mu_{10}^2\mu_{02} - \mu_{01}^2\mu_{20} + 2\mu_{01}\mu_{10}\mu_{11}\right), \nonumber
\end{align}
where $|A|$ denotes the determinant of any square matrix $A$. Observing that\\  $\displaystyle \int s^L(z,x,s,t,h_1,h_2) dF_{Y,X,T}(y,z,s) =1$, it follows that $\tilde{l}(x,t)$ in \eqref{eq:SPCL:EST4} corresponds to a localized Fr\'echet mean w.r.t. the Euclidean metric $d_E(a,b) := |a-b|$, 
\begin{align}\label{eq:SPCL:EST5b}
\tilde{l}(x,t) &=\underset{y \in  \reals }{\text{argmin  }}\E{s^L(X,x,T,t,h_1,h_2)d_E^2(Y,y)}.
\end{align}
The minimizer $\tilde{l}(x,t) $ can be viewed as a smoothed version of the true regression function, and can therefore be treated as an intermediate target.

This locally weighted Fr\'echet mean in \eqref{eq:SPCL:EST5b} can be readily generalized to the case of an $\Omega$-valued stochastic process $Y(t): t \in \T$, where $\Omega$ denotes a separable metric space, by retaining the same weights and replacing the Euclidean metric $d_E$ by $d$. This leads to the  intermediate populaton-level quantity, as is given below by model \eqref{eq:NPM:MODEL}.

In the context of nonparametric CORE, we thus define an intermediate function $\tilde{l}_\oplus(x,t)$  as a localized weighted Fr\'echet mean at the chosen points $(x,t),$ where
\begin{align}\label{eq:NPM:MODEL}
\tilde{l}_\oplus(x,t) := \underset{\omega \in \Omega}{\text{argmin }} \tilde{L}_\oplus(\omega,x,t),
\quad \text{where } \tilde{L}_\oplus(\omega,x,t) := \E{s^L(X,x,T,t,h_1,h_2)d^2(Y,\omega)}. 
\end{align}
Here  $s^L$ is as in \eqref{eq:SPCL:EST5a} and captures the local dependence of the response on the predictor.  
Minimizing the intermediate objective $\tilde{L}_\oplus(\omega,\cdot,\cdot)$ in \eqref{eq:NPM:MODEL} turns out to be  approximately the same as minimizing the final objective $M_\oplus(\omega)$ in \eqref{eq:CORE}. Finally, we propose an estimator for the intermediate target based on the plug-in estimates of the auxiliary parameters (see \eqref{eq:SPCL:EST3}) by their corresponding empirical estimates as follows.
Define
\begin{align}\label{eq:NPM:EST1}
\hat{\mu}_{jk} := \frac{1}{n}\displaystyle \sum_{i=1}^{n} \displaystyle \frac{1}{n_i} \sum_{l=1}^{n_i} K_{h_1,h_2}(X_{il}-x, T_{il}-t)(X_{il}-x)^j(T_{il}-t)^k,
\end{align}
\begin{align} 
\hat{\Sigma} & =
\begin{bmatrix}
\hat{\mu}_{00} & \hat{\mu}_{10} & \hat{\mu}_{01} \label{eq:3.99} \\
\hat{\mu}_{10} & \hat{\mu}_{20} & \hat{\mu}_{11} \\
\hat{\mu}_{01} & \hat{\mu}_{11} & \hat{\mu}_{02}
\end{bmatrix}, 
\quad 		\hat{\sigma}^2_0 = |\hat{\Sigma}|,	\quad 		N= \sum_{i=1}^{n}n_i ,\\
[\hat{\nu}_1, \hat{\nu}_2, \hat{\nu}_3]  &= \frac{1}{ \hat{\sigma}^2_0}\left[\hat{\mu}_{20}\hat{\mu}_{02}- \hat{\mu}_{11}^2, \quad\hat{\mu}_{01}\hat{\mu}_{11} -  \hat{\mu}_{02}\hat{\mu}_{10}, \quad \hat{\mu}_{10}\hat{\mu}_{11} - \hat{\mu}_{20}\hat{\mu}_{01}\right] \label{eq:3.100},\\
\hat{s}^L_{il}  (x,t,h_1,h_2) &= K_{h_1,h_2}(X_{il}-x, T_{il}-t) \left[ \hat{\nu}_1 + \hat{\nu}_2 (X_{il}-x) +\hat{\nu}_3 (T_{il}-t) \right] \label{eq:3.101}.
\end{align}

Plugging in the above empirical estimates we obtain the local Fr\'echet regression estimate 
\begin{align}\label{eq:NPM:EST2}
\hat{l}_\oplus(x,t) := \underset{\omega \in \Omega}{\text{argmin }} \hat{L}_\oplus(\omega,x,t), \quad   \hat{L}_\oplus(\omega,x,t) := \frac{1}{n}\displaystyle \sum_{i=1}^{n} \displaystyle \frac{1}{n_i} \sum_{l=1}^{n_i} \hat{s}^L_{il} (x,t,h_1,h_2)d^2(Y_{il},\omega).
\end{align}
Under suitable assumptions the bias introduced by changing the true target in \eqref{eq:CORE} to the intermediate target in \eqref{eq:NPM:MODEL}, given by $d(m_\oplus(\cdot,\cdot),\tilde{l}_\oplus(\cdot,\cdot))$, converges  to $0$ as the bandwidths $h_1,h_2 \rightarrow 0$. In addition  the stochastic term
$d(\hat{l}_\oplus(\cdot,\cdot),\tilde{l}_\oplus(\cdot,\cdot))$, converges to $0$ in probability, 
which then yields the  convergence of the proposed plug-in estimator in \eqref{eq:NPM:EST2} to the true target model in \eqref{eq:CORE}. To establish  this, we  
require the following assumptions, which are similar to  assumptions in  \cite{pete:mull:19}.

	\begin{enumerate}[label = (A\arabic*)]
	\item \label{A1}The kernel $K$ is symmetric around zero, with $|K^\gamma_{jk}| = | \int K^\gamma(u,v)u^kv^j du dv \displaystyle| < \infty \text{ for all } j,k =0,1,\dots, 6$ and $\gamma = 0,1,2.$ Also there is a common bandwidth parameter $h>0,\, h \rightarrow 0, \, nh \rightarrow \infty$  as $n \rightarrow \infty$, such that $h_1, h_2 \sim h.$ 
	
	\item \label{A2} The marginal density $f_{(X,T)}(x,t)$  and the conditional density $f_{(X,T)|Y}(x,t,y)$ exist and  are twice continuously differentiable with uniformly bounded derivatives as a bivariate function of $(x,t)$, the latter for all $y,$ for any given realization of $T=t,$ $X(T) = x,$ and $Y(T) =y.$
	
	\item \label{A3}The Fr\'echet means $m_\oplus(x,t), \, \tilde{l}_\oplus(x,t), \, \hat{l}_\oplus(x,t)$ exist and are unique for any given points $(x,t),$ and for any $\epsilon>0,$ 
	\[
	\underset{d(\omega,m_\oplus(x,t))>\epsilon}{\inf}\ M_\oplus(\omega,x,t) > M_\oplus(m_\oplus(x,t),x,t).
	\]
	
	\item \label{A4} For any $\epsilon>0$,
	\begin{align*}\underset{n}{\liminf}\underset{d(\omega,m_\oplus(x,t))>\epsilon}{\inf} \left(M_\oplus(\omega,x,t)- M_\oplus(m_\oplus(x,t),x,t)\right) >0,\\
	\underset{d(\omega,\tilde{l}_\oplus(x,t))>\epsilon}{\inf} \left(\tilde{L}_\oplus(\omega,x,t)- \tilde{L}_\oplus(\tilde{l}_\oplus(x,t),x,t)\right) >0.
	\end{align*}
	\item \label{A5} There exist constants $\eta_1 >0, \, C_1 >0$, $\text{with } d(\omega,m_\oplus(x,t)) < \eta_1$ such that
	\[
	M_\oplus(\omega,x,t) - M_\oplus(m_\oplus(x,t),x,t) \geq C_1d(\omega,m_\oplus(x,t))^2.
	\]
	
	\item  \label{A6} There exist $\eta_2 >0,  \, C_2 >0$, $\text{with } d(\omega,\tilde{l}_\oplus(x,t)) < \eta_2$ such that
	\[
	\underset{n}{\text{lim inf }} \left[ \tilde{L}_\oplus(\omega,x,t)  - \tilde{L}_\oplus(\tilde{l}_\oplus(x,t),x,t) \right]  \geq C_1d(\omega,\tilde{l}_\oplus(x,t))^2.
	\]
	
	\item \label{A7} Denoting the ball of radius $\delta$ centered at $m_\oplus(x,t)$ by $\mathcal{B}_\delta(m_\oplus(x,t)) \subset \Omega$ and  its covering number using balls of size $\epsilon$ as $N(\epsilon, \mathcal{B}_\delta(m_\oplus(x,t)),d)$,
	\[
	\int_{0}^{1} \sqrt{1+ \log N(\delta\epsilon, \mathcal{B}_\delta(m_\oplus(x,t)),d)} d\epsilon = O(1) \text{ as} \delta \rightarrow 0.
	\]
	\setcounter{cntr}{\value{enumi}}
\end{enumerate}
Assumptions \ref{A1}-\ref{A2} are necessary to show that the intermediate objective function $\tilde{L}_\oplus$ is a smoothed version of the true objective function $M_\oplus$. These are assumptions akin to the ones made in \cite{pete:mull:19} and are common in the nonparametric regression literature. Assumption  \ref{A3} is regarding the existence and uniqueness of the Fr\'echet means. The existence of the Fr\'echet means depends on the nature of the space, as well as the metric considered. For example, in case of Euclidean responses the Fr\'echet means coincide with the usual means for random vectors with finite second moments. In case of Riemannian manifolds the existence, uniqueness, and the convexity of the center of mass is guaranteed \citep{afsa:11,penn:18}. In a space with a negative or zero curvature, or in a Hadamard space  unique Fr\'echet means are also shown to exist \citep{bhat:03,bhat:05,patr:15,stur:03,kloe:10}.

Corresponding to each space equipped with a suitable metric, the computational challenge to find the Fr\'echet means could be different. In many cases, the key idea to compute the weighted Fr\'echet means reduces to solving a constrained quasi-quadratic optimization problem and projecting back into the solution space. For a wide class of objects such as distributions, positive semi-definite matrices, networks, and Riemannian manifolds among others, the unique solution can be found analytically (see Propositions $1$ and $2$ in the Appendix of \cite{pete:mull:19}), and is  not computationally difficult to obtain.  

Assumptions  \ref{A3}-\ref{A4} are commonly invoked  to establish consistency of an M-estimator such as $\hat{m}_\oplus(x,t)$, where one uses  the weak convergence of the empirical process $\hat{L}_\oplus$ to $\tilde{L}_\oplus$, which in turn converges smoothly to $M_\oplus$. Assumptions  \ref{A5}-\ref{A6} relate to  the curvature of the objective function and are needed to control the behavior of $\tilde{L}_\oplus- M_\oplus$ and $\hat{L}_\oplus - \tilde{L}_\oplus$ respectively, near the minimum.
Assumption \ref{A7} gives a bound on the covering number of the object metric space and is satisfied by the common examples of random objects such as distributions, covariance matrices, networks and so on.

In the concurrent regression framework, an important feature of the predictor space is as follows: when $X(t) \in \reals$, for any given $t \in \T$, the set  $\{\left(t,X(t)\right): t \in \T\}$ is a one-dimensional manifold $\mathcal{M}$ embedded in the ambient space $\reals^{2}$. This is an inherent property of the whole predictor space, irrespective of the dimension (possibly $p>1$) or the structure of $X(t).$ In our case, this reduces the effective dimension of the predictor space from two to one, i.e., the observed data  $(T_{il},X_{il})$  take values on this $1$-dimensional manifold embedded in $\reals^{2}$.  Note that this does not contradict our assumptions regarding the existence of the joint densities, $f_{X,T}$ (Section~\ref{model}).

Denoting  by $\mathcal{B}^{(k)}_r(a) \subset \reals^k$ a ball in $\reals^k$ with center $a \in \reals^k$ and radius $r>0$, for any $t\in \T$ and $x = X(t),$ the center of the ball $\mathcal{B}^{(2)}_h(x,t)$ is situated on the manifold $\mathcal{M}$. The following assumptions ensure that  the predictors are  dense on  $\mathcal{M}$. 
	\begin{enumerate}[label = (A\arabic*)]
	\setcounter{enumi}{\value{cntr}}
	\item 	 \label{A8}  Assume that for any $t \in \mathcal{T}$, the number of sample points outside balls $\mathcal{B}^{2}_h(x,t)$  is bounded and the he following asymptotic irrelevance condition hold.\\
	$\E{{K^\gamma \left(\frac{X-x}{h}, \frac{T-t}{h}\right) \ind{ (X(T),T) \notin  \mathcal{B}^{2}_h(x,t)} (X-x)^j (T-t)^k}} = O(h^{1+j+k}),$ for $\gamma = 0,1,2$, where  $\ind{z \notin A}$ denotes the indicator function for an element $z$ not belonging to a set $A$. 
	\item 	\label{A9} The density $f_T(\cdot)$ of $T$ is bounded away from $0$ the expected number of sample points falling inside a ball $\mathcal{B}^{2}_h(x,t)$ of radius $h$ centered at $(x,t)$ for any $t\in \T$ and $x = x(t)\in \reals$ is proportional to $h,$ i.e., for some  constant $c_{t}>0,$  
	$
	P((X_{il},T_{il}) \in \mathcal{B}^{2}_h(x,t)) = c_t h.
	$
\end{enumerate}
Assumptions akin to \ref{A8} are encountered in local polynomial regression \citep{bick:07, fan:96} to facilitate enough sample points to ensure estimation accuracy of the proposed methods. In particular it holds for a kernel $K$ with exponential tails.
Assumption \ref{A9} concerns the existence of a local ``chart'' or homeomorphism from a neighborhood in the predictor space $\reals^{2}$ to a ball in $\reals$, along the curve $\{(t,X(t)): t\in \T\}.$ 
This manifold structure of the predictor space is crucial to show that the rate of convergence corresponds to that for  $1-$dimensional  predictors even though the predictor dimension is $\reals^{2}$. For a generalization of the nonparametric CORE, where $X(t) \in \reals^p,$ for $p>1$ and for any $t\in \T,$ this observation still holds true and can be used to reduce the effective predictor dimension by one. 

The following propositions demonstrate that, while we have a \emph{two dimensional} predictor $(X,T)$, the rate of convergence of the proposed estimator still corresponds to the known optimal rate for a nonparametric regression with a \emph{one-dimensional} predictor. A similarly reduced rate of convergence is obtained for a  $p$-dimensional Euclidean predictor $X$. The reason that the effective predictor dimension is $p$ and not $(p+1)$ is the manifold constraint . Proposition~\ref{Proposition1} shows that the bias introduced by changing the concurrent object regression model $m_\oplus(\cdot,\cdot)$ in~\eqref{eq:CORE} to the intermediate nonparametric version of the CORE model $\tilde{l}_\oplus(\cdot,\cdot)$ in~\eqref{eq:NPM:MODEL} is negligible, as the bandwidth parameter for the bivariate kernel is chosen sufficiently small. Proposition~\ref{Proposition2} is about the stochastic convergence of the nonparametric CORE estimator $\hat{l}_\oplus(\cdot,\cdot)$ in~\eqref{eq:NPM:EST2}.  
\begin{Proposition}\label{Proposition1}
	Under the regularity assumptions \ref{A1}-\ref{A6}, for any given points $t\in \T$ and $x = X(t) \in \reals,$ 
	\begin{equation*}
	d(m_\oplus(x,t), \tilde{l}_\oplus(x,t)) = O\left(h^2\right),\ \text{ as }  h = h_n \rightarrow 0 \text{ and }nh \rightarrow \infty \text{ where } h \text{ is as in } \ref{A1}.
	\end{equation*}
\end{Proposition} 

\begin{Proposition}\label{Proposition2}
	Under the  regularity assumptions  \ref{A1}-\ref{A9}, for any given points $t\in \T$ and $x = X(t) \in \reals,$ 
	\begin{equation*}
	d(\hat{l}_\oplus(x,t), \tilde{l}_\oplus(x,t)) = O_p ((nh)^{-1/2}),\ \text{ as }  h = h_n \rightarrow 0 \text{ and }nh \rightarrow \infty \text{ where } h \text{ is as in }  \ref{A1}.
	\end{equation*}
\end{Proposition}
In general, the rate of convergence is dictated by the local geometry of the object space near the minimum as quantified in \ref{A4}-\ref{A6}. The derivations for the pointwise results 
rely on tools from  the theory of M-estimation.
Combining these two results leads to the overall rate of convergence of the nonparametric CORE estimator.
\begin{Theorem} \label{thm1}
	Under the regularity conditions \ref{A1}-\ref{A9}), 
	\begin{equation*}
	d(m_\oplus(x,t),\hat{l}_\oplus(x,t)) = O_p\left( h^2 + (nh)^{-\frac{1}{2}}\right),  \text{ as } h = h_n \rightarrow 0  \text{ and } nh \rightarrow \infty.
	\end{equation*}
\end{Theorem} 
Under the Assumptions  \ref{A1}-\ref{A9}, if we consider a sequence of bandwidths of the form $h=n^{-\gamma}$, the optimal choice for $\gamma$ that minimizes the mean square error is obtained for $\gamma^{\ast} = 1/5$ and the resulting rate of convergence is  $d(m_\oplus(x,t),\hat{l}_\oplus(x,t)) = O_p\left(n^{-2/5} \right) $.

The nonparametric CORE model and assumptions considered so far are developed for the case $X(\cdot) \in \reals$. For instance, the kernel is assumed to be bivariate, and the weights $s^L$ in~\eqref{eq:SPCL:EST5a} and their estimates in~\eqref{eq:3.101} accommodate a real-valued predictor process. The theory can be generalized for $p>1,$ however, there are practical limitations, including the curse of dimensionality, multiple bandwidth choices, and one has to account for correlation and differences in scale between the components of $X(\cdot).$
Under more stringent modeling assumptions some of these issues can be avoided by a modeling approach that extends the notion of linear relationship to the $X$ direction and this will be discussed next section.

\section{Partially global concurrent object regression}
\label{PGM}
In the Euclidean case, a  well-established  alternative to  nonparametric concurrent regression is a global/linear varying coefficient model, where for each fixed time a linear regression of $Y(\cdot)$ on $X(\cdot)$ is assumed.  This linear regression relation can be described by a global weight function applied to the covariate $X(\cdot)$. This can then be adapted for  the case where responses are random objects, by constructing conditional Fr\'echet means with this same weight function  \citep{pete:mull:19}, all while assuming nonlinear  dependence between $Y(T)$ and $T$. 
As before, we first study  the special case of a Euclidean response and then express the CORE function in \eqref{eq:CORE} as an intermediate target expressed as a weighted Fr\'echet mean, the weights being globally linear in the $X$-direction and locally linear in the $T$- direction. The partially linear dependence in the $X-$ direction imposes a more structural model than the general conditional Fr\'echet mean defined in~\eqref{model}.
This leads to the proposed partially global concurrent object regression model, with the Euclidean predictor $X(\cdot) \in\reals^p,$ ($p\geq 1$) and object response $Y(\cdot) \in \Omega,$ at the given points $T=t$ and $X(T)= x$ as
\begin{align}\label{eq:PGM:MODEL}
\tilde{g}_\oplus(x,t) = \underset{\omega \in \Omega}{\text{argmin}} \, \tilde{G}_\oplus(\omega,x,t), \quad 
\text{where} \, \tilde{G}_\oplus(\omega,x,t) := \E{s^G(X,x,T,t,h)d^2(Y,\omega)}.
\end{align}

Here the weight function $s^G$ is given by
\begin{align}
s^G(z,x,s,t,h) = s_1(z,x,s,t,h) + s_2(s,t,h), \label{eq:PGM:EST1}
\end{align} 
with $s_1(z,x,s,t,h) := K_h(s-t)\left[(z-\mu_X(t))^\intercal  \Sigma^{-1}_{20} (x-\mu_X(t))\right]$, where $\mu_X(t) = \E{X(t)} =$ $\E{X|T=t},$
and $s_2(s,t,h) := \frac{1}{\sigma_0^2} K_h(s-t)\left(\mu_{02} - (s-t)\mu_{01} \right).$ 
For the explicit derivation of the weight function $s^G$, motivated from the special case of time-varying Euclidean responses.
Here $s^G$ encapsulates  the dependence of the response on the predictors, where  the dependence is global in the direction of the covariate $X$,  while it is local in the $T$ direction, which is reflected in the   two parts $s^G(z,x,s,t,h) =s_1(z,x,s,t,h) + s_2(s,t,h)$.  

Observe that $s_1(\cdot,\mu_X(t),\cdot,\cdot,\cdot) = 0$,  that is, the regression model reduces to a nonparametric regression model with the only predictor $T$ when $x= \mu_X(t)$. We see that
$
\int s_1(z,x,s,t,h) \ dF_{(X,T)}(z,s) =0.
$
Also, under mild assumptions (Assumption \ref{B1}  in the Supplementary Material Section~\ref{Appendix:A3}) on the kernel $K_h(\cdot)$ and the smoothness of marginal and conditional densities $f_{(X,T)}$ and $f_{(X,T)|Y}$ we can show that
$
\int s_2(s,t,h) \ dF_{X,T|Y}(z,s,y)= \frac{dF_{X,T|Y}(z,s,y)}{dF_{X,T}(z,s)} +O\left(h^2\right).
$
Thus we may view $\tilde{G}_\oplus$ as a smoothed version of $M_\oplus$ as the bandwidth parameter $h=h_n \rightarrow 0$ 

Finally, we propose a plug-in estimate for the partially-global regression model $g_\oplus$ in \eqref{eq:PGM:MODEL}. For this purpose we define the preliminary estimates of the auxiliary parameters as follows
\begin{align}
\hat{\mu}_{0j} &:= \frac{1}{n}\displaystyle \sum_{i=1}^{n}  \frac{1}{n_i}\displaystyle \sum_{l=1}^{n_i}K_h(T_{il}-t)(T_{il}-t)^j, \label{eq:PGM:EST3}\\
\hat{\Sigma}_{2j} &:= \frac{1}{n}\displaystyle \sum_{i=1}^{n} \frac{1}{n_i} \displaystyle \sum_{l=1}^{n_i} K_h(T_{il}-t)(T_{il}-t)^j(X_{il}-\hat{\mu}_X(t)) (X_{il}-\hat{\mu}_X(t))^T ,\\
\hat{\sigma}_0^2 &:= \hat{\mu}_{02}\hat{\mu}_{00} - \hat{\mu}_{01}^2.
\label{eq:PGM:EST4}
\end{align}
The mean function $\mu_X(\cdot)$ for the predictor process $X(\cdot)$ is estimated by $\hat{\mu}_X(\cdot)$
by smoothing the aggregated data $(T_{il},X_{il})$ $i = 1,\dots,n,$ $j = 1,\dots,n_i,$ with local linear fitting~\citep{yao:05}. We then calculate  empirical weights using the auxiliary parameters from above as
\begin{align}
\hat{s}^G_{il}(x,t,h)= K_h(T_{il}-t) \left[\left(X_{il}-\hat{\mu}_X(t)\right)^T \hat{\Sigma}^{-1}_{20} \left(x-\hat{\mu}_X(t)\right) + \frac{1}{\hat{\sigma}_0^2}\left(\hat{\mu}_{02}-(T_{il}-t)\hat{\mu}_{01}\right)\right]. \label{eq:PGM:EST5}
\end{align}
The proposed partially global concurrent object regression (CORE) estimate is given by
\begin{eqnarray}\label{eq:PGM:EST6}
\hat{g}_\oplus(x,t) = \underset{\omega \in \Omega}{\text{argmin}} \, \hat{G}_\oplus(\omega,x,t), \,
\text{ where  } \hat{G}_\oplus(\omega,x,t) = \frac{1}{n} \sum_{i=1}^{n} \Bigl( \frac{1}{n_i} \sum_{l=1}^{n_i} \hat{s}^G_{il}(x,t,h)d^2(Y_{il},\omega) \Bigr).
\end{eqnarray} 
Further motivation of this approach,  starting from the case of Euclidean responses, can be found in the Supplementary Material Section~\ref{Appendix:A2}.
We show consistency with an optimal rate for the proposed model to the target CORE function in \eqref{eq:CORE} under assumptions  \ref{B1}-\ref{B6} (see the Supplementary Material Section~\ref{Appendix:A3}), which are similar to the assumptions  \ref{A1}-\ref{A6} in Section~\ref{NPM}.
\begin{Proposition}\label{Proposition3}
	Under the assumptions \ref{B1}-\ref{B3}, for any given points $t\in \T$ and $x = X(t) \in \reals^p,$ 
	\begin{equation*}
	d(m_\oplus(x,t), \tilde{g}_\oplus(x,t)) = O\left(h^2\right), \text{ as } h = h_n \rightarrow 0.
	\end{equation*}
\end{Proposition}

\begin{Proposition}\label{Proposition4}
	Under the assumptions \ref{B1}-\ref{B6}, for any given points $t\in \T$ and $x = X(t) \in \reals^p,$ 
	\begin{equation*}
	d(\hat{g}_\oplus(x,t), \tilde{g}(x,t)) = O_p((nh)^{-1/2}), \text{ as } h = h_n \rightarrow 0 \text{ and } nh \rightarrow \infty.
	\end{equation*}
\end{Proposition}
Combining these two results leads to the pointwise consistency for the partially global CORE estimator as follows:
\begin{Theorem}\label{thm2}
	Under \ref{B1}-\ref{B6},
	\begin{equation*}
	d(\hat{g}_\oplus(x,t), m_\oplus(x,t)) = O_p(h^2 + (nh)^{-1/2}), \text{ as } h = h_n \rightarrow 0 \text{ and } nh \rightarrow \infty.
	\end{equation*}
\end{Theorem}
Comparing to the local rates of convergence for the Nonparametric CORE estimator, as proposed in Section~\ref{NPM}, the rates in Propositions~\ref{Proposition3} and \ref{Proposition4} are global in the predictor $X$ and remain unchanged even for a higher predictor dimension $p,\ p>1$. For $p=1$, both the estimators behave in a similar manner, however as $p$ increases the partially global estimator performs better in terms of the rate of convergence to the true CORE model in \eqref{eq:CORE}. While the above results are pointwise, a uniform convergence result in a compact interval in the $X$-direction also holds  for any given point in the $T$-direction, under slightly stronger  assumptions (see assumptions \ref{U1}-\ref{U4} in the Supplementary Material Section~\ref{Appendix:A3}). 
Denoting  the Euclidean norm on $\reals^p$ by $\lvert \lvert \cdot\rvert \rvert _E$, we obtain 
\begin{Theorem}\label{thm3}
	Under the assumptions \ref{U1}-\ref{U4}, 
	for any given $t \in \T$ and $M>0,\\ \text{ as } h = h_n \rightarrow 0 \text{ and } nh \rightarrow \infty,$
	\begin{equation*}
	\underset{||x||_E \leq M} {\sup}d(\hat{g}_\oplus(x,t), m_\oplus(x,t)) =
	O_p\left(h^2 + (nh)^{-1/2+ \delta}\right), \text{ \ for any } \delta >0.
	\end{equation*}
\end{Theorem}
The proofs require results from empirical process theory.

\section{Simulation studies}
\label{SIM}
\subsection{Distributional object responses}
\label{SIM:dens}
We illustrate the efficacy of the proposed methods through simulations, where the  space of distributions with the Wasserstein metric provides an ideal setting. We consider time-varying distributions on a bounded domain $\mathcal{T}$ as the response, $Y(\cdot)$, and they are represented by the respective quantile functions $Q(Y)(\cdot)$.  The time-varying Euclidean random variable $X(\cdot)$ is taken as the predictor. 
The random response is generated conditional on $(X(T),T)$, by adding noise to the true regression quantile
\begin{align}
Q(m_\oplus(x,t))(\cdot) &= \E{Q(Y)(\cdot)|X(t) = x, T= t}. 
\label{eq:SIMUL1}
\end{align}

Two different simulation scenarios are examined as we generate the distribution objects from location-scale shift families (see Table~\ref{table:sim}). In the first setting, the response is generated, on average, as a normal distribution with parameters that depend on $(T,X(T))$. For $T = t, \ X(T) =x$, the distribution parameters $\mu \sim N(\zeta(x,t), \nu_1)$ and $\sigma \sim Gamma\left(\frac{\eta^2(x,t)}{\nu_2}, \frac{\nu_2}{\eta^2(x,t)}\right)$ are independently sampled. The corresponding distribution is given by $Q(Y)(\cdot) = \mu + \sigma \Phi^{-1}(\cdot)$. Here, the relevant sub-parameters are chosen as $\nu_1 = 0.1,$  $\nu_2 = 0.1$, $\zeta(x,t) = 0.5 + 0.1x +0.1t^2$, and $\eta(x,t) = 0.5 + 0.1x + 0.1 \sin (10\pi t)$, and $\Phi(\cdot)$ is the standard normal distribution function. 

The second setting is slightly more complicated. The distributional parameter $\mu|(X(t)=x, T=t)$ is sampled as before and $\sigma = 0.1$ is assumed to be a fixed parameter. The resulting distribution is then ``transported'' in Wasserstein space via a random transport map $T$, that is uniformly sampled from the collection of maps $T_k(a) = a - \sin (ka)/|k|$ for $k \in {\pm 1, \pm 2}$. The distributions thus generated are not Gaussian anymore due to the transportation. Nevertheless, one can show that the Fr\'echet mean is exactly $ \mu + \sigma \Phi^{-1}(\cdot)$ as before.

\begin{table}[h]
	\centering
	\begin{tabular}{|l|l|}
		\hline
		Setting I &
		Setting II \\ \hline
		\begin{tabular}[c]{@{}l@{}}$Q(Y)(\cdot) = \mu + \sigma \Phi^{-1}(\cdot) $, \\ where \\ $\mu \sim N(\zeta(x,t), \nu_1) $\\ $\sigma \sim Gamma\left(\frac{\eta^2(x,t)}{\nu_2}, \frac{\nu_2}{\eta^2(x,t)}\right)$\end{tabular} &
		\begin{tabular}[c]{@{}l@{}}$Q(Y)(\cdot) = T \#(\mu +\sigma \Phi^{-1}(\cdot)) $, \\ where\\ $\mu \sim N(\zeta(x,t), \nu_1) $\\ $\sigma = 0.1,$ \\ $T_k(a) = a - sin(ka)/|a|,  k \in \{\pm1,\pm 2\}$\end{tabular} \\ \hline
	\end{tabular}
	\caption{Table showing two different simulation scenarios.}
	\label{table:sim}
\end{table}

To this end, we generated a random sample of size $n$ of time-varying response and predictors from the true models, where the $i^{\text{th}}$ sample was observed at $n_i$ random time points, incorporating measurement error as described in the two situations above. For simplicity, we chose $n_i =m$ to be equal for all subjects and consider the two cases with $n_i = 5$ and $n_i=20.$ Each such case was repeated for sample sizes $n = 100$ and $n = 1000.$ For a given $n_i$ and $n, $ we first sampled the time points $T_{il} \overset{i.i.d.}{\sim} Unif(0,1)$ for $\ l = 1, \dots n_i$ and $i=1,\dots,n.$
The predictor trajectories $X_i(\cdot)$ were generated as follows. The simulated processes $X$ had the mean function $\mu_X(t) = t + \sin (t)$, with covariance function constructed from $K= 10$ eigen functions, $\phi_1(t) = -\cos(\pi t/10)/\sqrt{5}$, and $\phi_j(t) = \sin((2j - 1)\pi t/10)/\sqrt{5},$ for $t \in [0,1],$ $j=2,\dots K.$ We chose $\lambda_1 = 1,$ $\lambda_2 = .7,$  to be the first two eigen values and $\lambda_j = (0.7)^{j-1}$ for $j=3,\dots,K$ as the remaining eigenvalues. The FPC scores $\xi_{ij}$s were generated from $N(0,\lambda_j)$ truncated on $[-6,6]$ for $j=1\dots,K.$ Using the Karhunen–Lo\`eve theorem, the predictor process is generated at the random time-points $T+{il}$ as $X_{i}(T_{il}) = \mu_X(T_{il}) + \sum_{j=1}^{K}\xi_{ij}\phi_j(T_{il})$ for $l=1\dots,n_i$ and $i=1\dots,n.$

For each of Setting I and II, $500$ Monte Carlo runs were executed for a combination of sample
sizes $n$ and $n_i$, including both sparse and dense designs. For the $r^{\text{th}}$ simulation,  $\hat{f}^r_\oplus(x,t)$ denoting the fitted distribution function, and $f^r_\oplus(x,t)$ denoting the simulated density objects, the utility of the estimation was measured quantitatively by the integrated squared errors
\begin{align}
\text{ISE}_r &= \int_{0}^{1} \int_{-6}^{6} d_W^2(\hat{f}^r_\oplus(x,t), f^r_\oplus(x,t)) dx dt,
\end{align}
where $d_W$ denotes the Wasserstein metric between two distributions.
\begin{figure}[!h]	
	\centering
	\includegraphics[width=.7\textwidth]{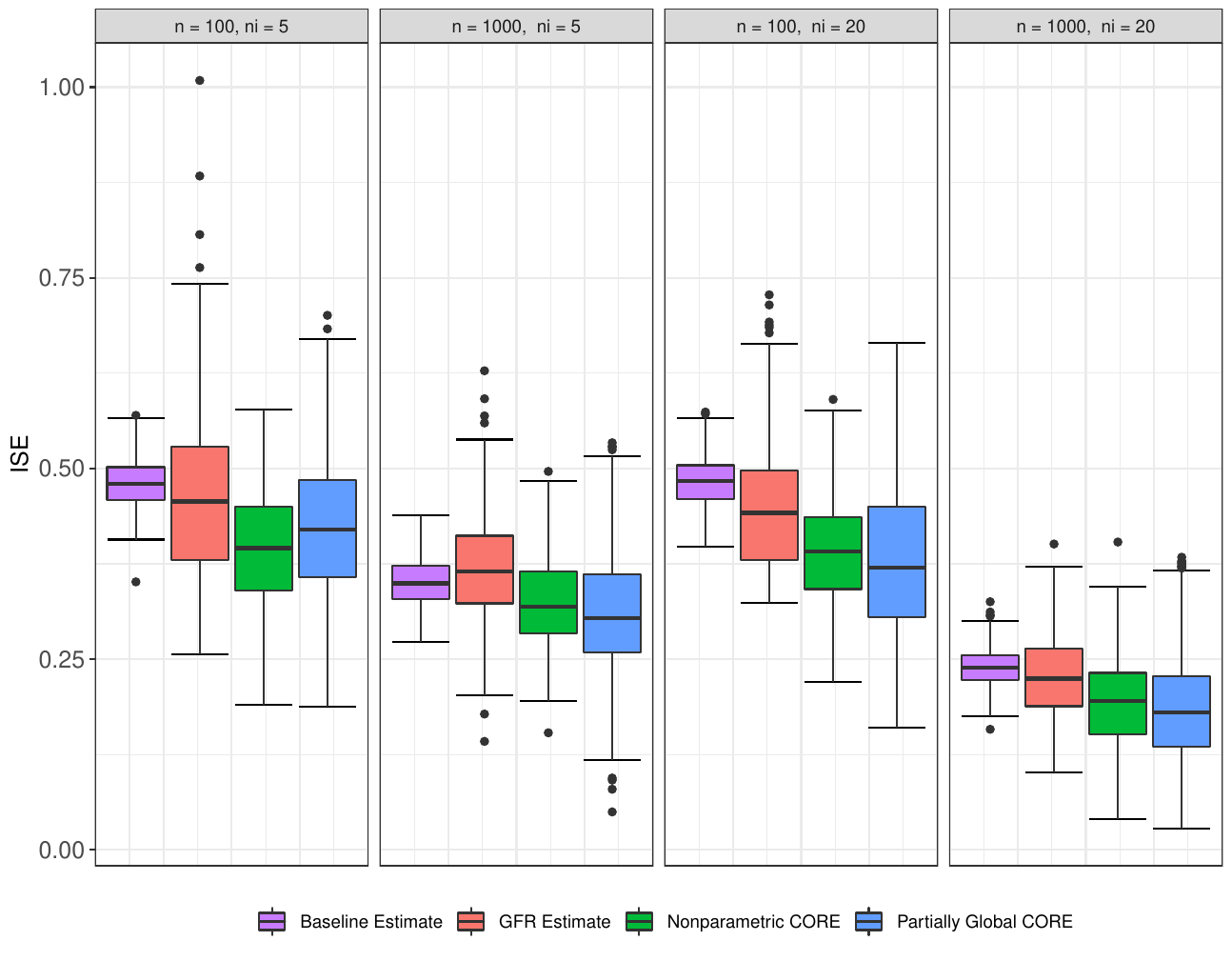}
	\caption{Boxplots of Integrated Squared Errors (ISE) over $500$ simulation runs and different sample sizes for density estimates resulting from partially global and nonparametric concurrent object regression (CORE), global Fr\'echet regression (GFR) and a baseline model in the simulation setting I, as described in Table~\ref{table:sim}.}
	\label{Fig:Sim:Boxplot}
\end{figure}	
We fitted both of the nonparametric and partially global concurrent object regression (CORE) models over a grid of points $x = x(t) \in [-6,6]$ and $t \in [0,1].$ The bandwidths for the estimation in both the settings were chosen over a grid of possible values using a cross validation criterion so as to minimize the average ISE for all simulations. For the $x-$ direction a grid of bandwidths $h_2 \in [n^{-1/5},3.18n^{-1/5}]$ was used for this purpose, while for the $t-$ direction a grid of bandwidths $h_1 \in [0.05n^{-1/5},0.265n^{-1/5}]$ was used. A truncated bivariate Gaussian product kernel and a truncated univariate Gaussian kernel were chosen to fit the nonparametric CORE and the partially global CORE methods, respectively.

In Setting I, the performances of the proposed CORE models were compared to a baseline linear concurrent model, which is mis-specified in our case. As such, since in the first setting we knew the finite-dimensional generating model, we computed the mean $\mu_i(T_{il})$ and standard deviation $\sigma_i(T_{il})$ of the distribution $Y_{il}$ and regressed each of them  linearly against the predictors $(X_{il},T_{il})$. The quantile functions for the baseline model was computed as $\hat{\mu}(x,t) + \hat{\sigma}(x,t) \Phi^{-1}(\cdot),$ where $\hat{\mu}(x,t)$ and $\hat{\sigma}(x,t)$ were the estimated mean and variance functions at $(x,t)$ using the fitted coefficients from the previous step. Clearly, the baseline concurrent model is mis-specified, but it highlights the fact that the proposed CORE models are the only applicable regression model, to the best of our knowledge, in the context of concurrent regression for distributional object responses. We also compared the performance of the CORE models to that of the global Fr\'echet regression (GFR) model~\citep{pete:mull:19} where $T$ and $X$ were used as a two-dimensional predictor, ignoring the inherent nested structure of the predictor space $(T,X(T))$. We observed a decrease in ISE for all the models as the sample size was increased, favorably  for the denser design with $n_i = 20$ (see Figure~\ref{Fig:Sim:Boxplot}). The CORE models outperformed both the baseline (mis-specified) model and the GFR model. Further, the partially global CORE had slightly lower ISE value than the nonparametric one, specially for denser designs. This is expected since in this simulation setting, the global model holds true in the $x-$ direction.
\begin{figure}[!htb]	
	\centering
	\includegraphics[width=.7\textwidth]{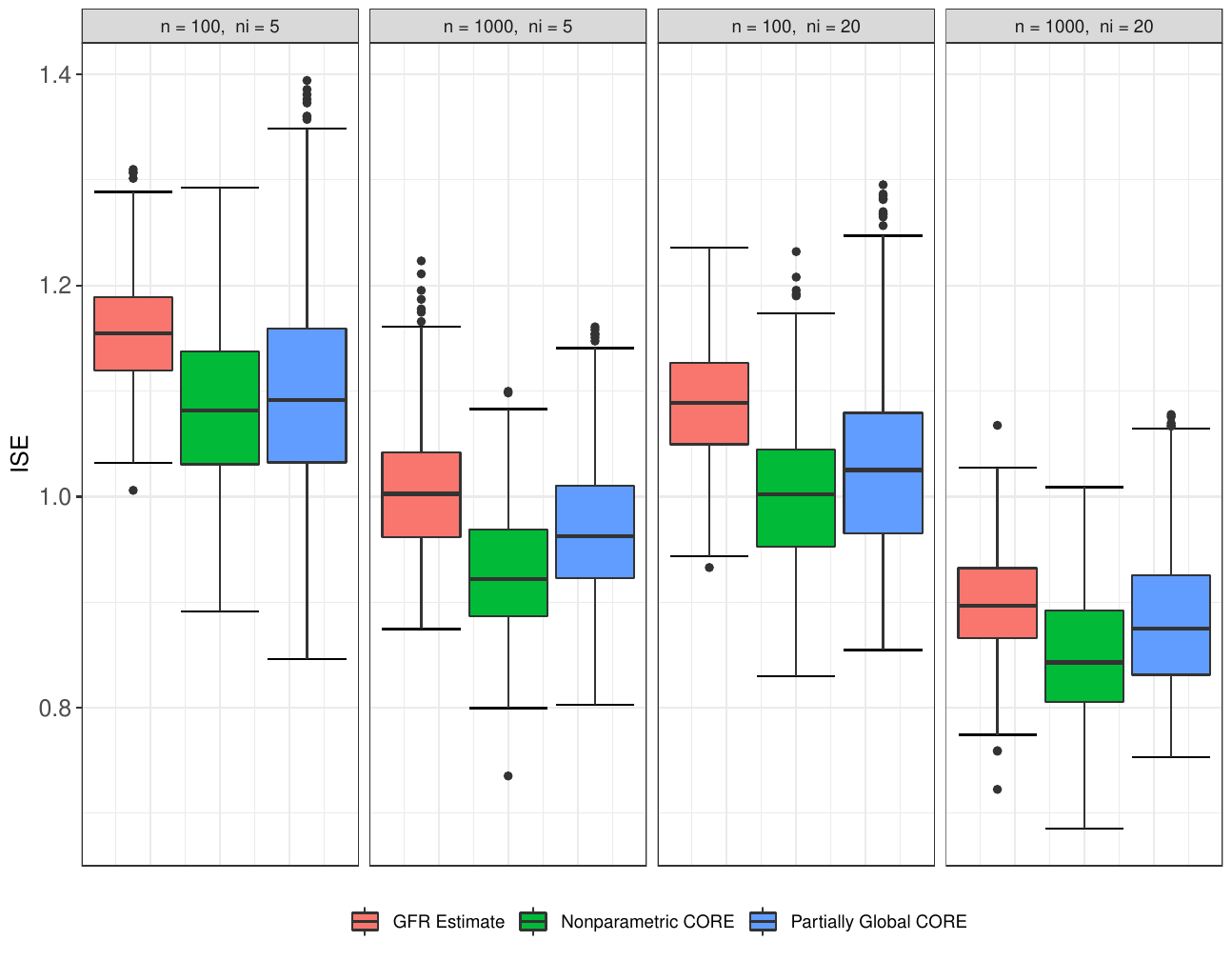}
	\caption{Boxplots of Integrated Squared Errors (ISE) for $500$ simulation runs and different sample sizes for density estimates resulting from  partially global and nonparametric concurrent object regression (CORE) and global Fr\'echet regression (GFR) for simulation setting II, as described in Table~\ref{table:sim}.}
	\label{Fig:Sim:Boxplot2}
\end{figure}	

In the second simulation setting, the baseline linear model is no longer admissible due to the random transportation step, thus the baseline model is dropped for the comparison purpose. However, we could still compare the performances among the two proposed CORE models and the GFR model. Both CORE methods performed in a similar manner and outperformed the GFR in all scenarios (see Figure~\ref{Fig:Sim:Boxplot2}). We again observed a decreasing pattern of the integrated squared errors for increasing sample sizes and denser designs, demonstrating the validity of the CORE models for this complex and time-varying regression setting. The nonparametric CORE performed better for a higher sample size. This is not unexpected since the data generating mechanism was non-linear and the partially global model assumes a linear dependence in the $x-$ direction. Further, the comparative performance of the partially global CORE method to that of GFR was studied for increasing the predictor dimension $p$ (see the Supplementary Material Section~\ref{Appendix:A6}). For the implementation of the GFR method again the nested structure of the predictor space $(T,X(T))$ was ignored, as such $T\in \reals$ and $X\in\reals^{p}$ were treated as a $p+1$ dimensional predictor input for the model.

\subsection{Object responses on a unit sphere}
\label{SIM:SPHE}
We next implemented our methodology when the responses lie on a Riemannian manifold object space - in particular we considered responses lying on the surface of a unit sphere $S^2$ in $\reals^3$ with the center being the origin. The geodesic distance between any two points $\omega_1$ and $\omega_2$ lying on the surface of the unit sphere $S^2$ is given by $d(\omega_1,\omega_2) = \arccos (\omega_1^T\omega_2).$ We considered the concurrent object regression function as follows
\[
m_\oplus(x,t) = ((1-\left(x/a\right)^2)^{1/2}\cos(\pi t),(1- (x/a)^2)^{1/2}\sin(\pi t), (x/a)),\ t\in(0,1),\  x\in (-a,a),\ a>0.
\]
We first generated the predictor process $(T_{il},X_i(T_{il}))$ as before (see Section~\ref{SIM:dens}) such that $T_{il} \in (0,1)$ and $X_i(T_{il}) \in (-a,a)$ with $a=6$ for $l=1,\dots,n_i,$ $i=1\dots,n.$ The response was then constructed as follows. A bivariate noise random vector was generated  on the tangent space $T_{m_\oplus(X_{il},T_{il})}(\Omega).$ To this end, we defined $\psi_{il} = \arcsin (T_{il})$ and $\theta_{il} = \pi T_{il}.$ An
orthonormal basis for the tangent space was denoted by $(b_{1il},b_{2il}),$ where
$b_{1il} = (\cos(\psi_{il}) \cos(\theta_{il}), \cos(\psi_{il}) \sin(\theta_{il}), -\sin(\psi_{il}))^T$ and $b_{2il} = (\sin(\theta_{il}), -\cos(\theta_{il}),0)^T.$
Adding a noise level $\sigma^2 = 0.1$,  bivariate random vectors $Z_{il} = c_{i1}b_{1il} + c_{i2}b_{2il}$ were computed, where $C_i = (c_{i1},c_{i2})^T \overset{i.i.d.}{\sim} N_2(0,\sigma^2I_2)$ with $\sigma^2 = 0.1.$ Finally, the response was constructed as
\[
Y_{il} = \cos\left( \lVert Z_{il}\rVert_E\right) m_\oplus(X_{il},T_{il}) + \sin \left( \lVert Z_{il}\rVert_E\right) \frac{Z_{il}}{\lVert Z_{il}\rVert_E},
\]
with $\lVert \cdot \rVert_E$ being the Euclidean norm. The simulation  steps produced a point $Y_{il}$ on the surface of the two-dimensional sphere with conditional Fr\'echet mean $m_\oplus(X_{il},T_{il})$ contaminated with a small level of noise.
\begin{figure}[!h]	
	\centering
	\includegraphics[width=.7\textwidth]{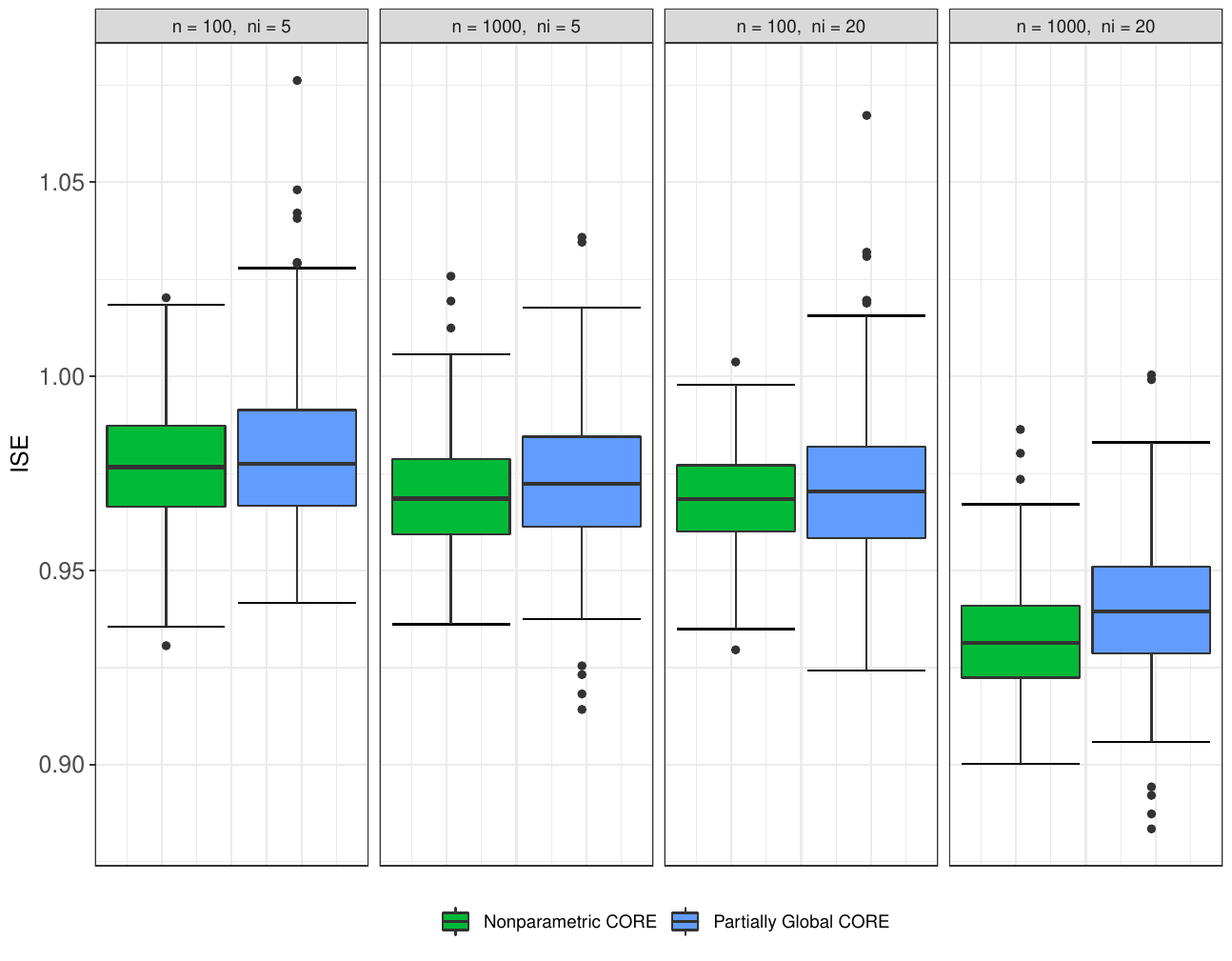}
	\caption{Boxplots of Integrated Squared Errors (ISE) for $500$ simulation runs and different sample sizes for object responses on the surface of the unit sphere $S^2$ resulting from  partially global and nonparametric concurrent object regression (CORE).}
	\label{Fig:Sim:Boxplot3}
\end{figure}	

We fitted the two concurrent object regression (CORE) models for the simulated data over a grid of points $x(t) =x \in (-6,6)$ and $t \in (0,1)$. For each of the CORE models,
$500$ Monte Carlo runs were implemented corresponding to combinations of sample
sizes $n$ and $n_i$, including both sparse and dense designs. For the $r^{\text{th}}$ simulation, at any given point  $(x,t),$ $Y^r_\oplus(x,t)$ and $\hat{Y}^r_\oplus(x,t)$ denoted the simulated and fitted objects on the surface of the unit sphere $S^2.$ The performance of the model was measured quantitatively by the integrated squared errors
\begin{align}
\text{ISE}_r &= \int_{0}^{1} \int_{-6}^{6} d_g^2(\hat{Y}^r_\oplus(x,t), Y^r_\oplus(x,t)) dx dt,
\end{align}
where $d_g$ denotes the geodesic distance on a unit sphere $S^2.$ The bandwidths for the estimation were chosen using a cross validation criterion so as to minimize the average ISE over all simulations, and a truncated Gaussian kernel was chosen.
Figure~\ref{Fig:Sim:Boxplot3} shows that, as before, with an increasing sample size and denser design the average ISE reduces for both nonparametric and partially global CORE models. Further, in this highly nonlinear simulation scenario, the nonparametric CORE performs better in terms of lower estimation error, specially for a larger sample size and dense design.
\section{Data illustrations}
\label{DATA}
\subsection{Brain connectivity in Alzheimer's disease} \label{DATA:fMRI}
Modern functional Magnetic Resonance Imaging (fMRI) methodology has made it possible to study structural elements of the brain and identify brain regions or cortical hubs that  exhibit similar behavior, especially when subjects are in the resting state \citep{allen:14,ferr:busa:13}. 

In resting state fMRI, a time series of Blood Oxygen Level Dependent (BOLD) signal is observed for the seed voxels in selected functional hubs. For each hub, a seed voxel is identified as the voxel whose signal has the highest correlation with the signals of nearby voxels. 

Alzheimer's Disease has been found to have associations with anomalies in functional integration of brain regions and target regions or hubs of high connectivity in the brain \citep{damo:12,zhang:10}.

Data used in the preparation of this article were obtained from the Alzheimer’s Disease Neuroimaging Initiative (ADNI) database (\url{adni.loni.usc.edu}). For up-to-date information, see \url{www.adni-info.org}. Brain image-scans for subjects in different stages of the disease were available, along with other relevant information such as age, gender, and total cognitive score, recorded on the same date as the scan. 

For this analysis, subjects aged from $55$ to $90$ years and belonging to either of the Alzheimer's Disease (AD) or Cognitive Normal (CN) patient groups were considered. After removing the outliers, the number of image scans recorded were $174$ and $694$, respectively, for the $78$ AD subjects and $371$ CN subjects who participated in the study. To confirm that the age intervals across the two groups are comparable, we first performed a Kruskal- Wallis test for the null hypothesis of equal age distributions of the two groups, which resulted in a p-value of $0.92$, indicating no evidence for systematic age differences.

BOLD signals for $V= 10$ brain seed voxels for each subject were extracted. The $10$ hubs where the voxels are situated are labeled as follows: LMF and RMF (left and right middle-frontal), LPL and RPL (left and right parietal), LMT and RMT (left and right middle temporal), MSF (medial superior frontal), MP (medial prefrontal), PCP (posterior cingulate/precuneus) and RS (right supramarginal), as discussed in \cite{buck:09}. 

The preprocessing of the BOLD signals was implemented by adopting the standard procedures of slice-timing correction, head motion correction and normalization and other standard steps. The signals for each subject were recorded over the interval $[0, 270]$ (in seconds), with $K=136$ measurements available at 2 second intervals. From this the temporal correlations were computed to construct the connectivity correlation matrix, also referred to as the Pearson correlation matrix in the area of fMRI studies. 

The observations were available sparsely at random time-points, such that the $i^{\text{th}}$ subject is observed at $n_i$ time-points, $n_i$ varying from a minimum of $1$ to a maximum of $7$. The inter-hub connectivity Pearson correlation matrix $Y_{il}$, for the $i^{\text{th}}$ subject observed at age $T_{il}$ (measured in years), has the $(q,r)^{\text{th}}$ element 
\begin{align}
(Y_{il})_{qr} = \frac{\sum_{p=1}^{K}(s_{ipq} - \bar{s}_{iq}) (s_{ipr} - \bar{s}_{ir})}{\left[\left(\sum_{p=1}^{K}(s_{ipq} - \bar{s}_{iq})^2\right) \left(\sum_{p=1}^{K}(s_{ipq} - \bar{s}_{iq})^2 \right) \right]^{1/2}},
\label{eq:PEARSON:CORR}
\end{align}
where $s_{ipq}$ is the $(p,q)^{\text{th}}$ element of the signal matrix for the $i^{\text{th}}$ subject and $\bar{s}_{iq} := \frac{1}{K}\sum_{p=1}^{K} s_{ipq}$ is the mean signal strength for the $q^{\text{th}}$ voxel.

For Alzheimer's disease trials, ADAS-Cog-13 is a widely-used measure of cognitive performance. It measures impairments across several cognitive domains that are considered to be affected early and characteristically in Alzheimer's disease \citep{scara:18, kuep:18}. It is important to note that higher scores are associated with more serious cognitive deficiency. 
To study how functional connectivity in the brain varies with the total cognitive score for subjects at different ages, we applied the CORE models. It is known that age affects both functional connectivity in the brain and total cognitive score so that the relation of cognitive deficits with brain connectivity likely changes with age. 

We implemented a time-varying or concurrent regression framework with the Pearson correlation matrices in \eqref{eq:PEARSON:CORR} as time-varying object responses, residing in the metric space of correlation matrices equipped with the Frobenius norm, and total cognitive scores as real-valued covariates, changing with time (age in years). Specifically, we fitted the nonparametric CORE in \eqref{eq:NPM:MODEL} separately for the AD and CN subjects over different output points for age $t$ and total cognitive score $x$.
The bandwidths in the local fits for both the age and total cognitive score directions were chosen satisfying a leave-one-out cross validation criterion with a bivariate Normal kernel function, which led to the bandwidths in Table~\ref{Table:Bandw:NPM}.


\begin{table}[h!]
	\centering
	\begin{tabular}{  ccc } 
		\hline
		&AD &  CN \\ 
		\hline
		$h_1$ & 3.95 & 3.64 \\ 
		\hline
		$h_2$ & 3.78 & 2.43 \\ 
		\hline
	\end{tabular}
	\caption{Bandwidths used in the nonparametric CORE model for the AD and CN subjects, here $h_1$ is the bandwidth for age and $h_2$ for total cognitive score.}
	\label{Table:Bandw:NPM}
\end{table}

We fitted the proposed model at the $x =10\%,\ 50\%,\ $ and $90\%$ quantile values in the total cognitive score direction, where higher total score means larger cognitive impairment. We find that for higher scores and thus increased cognitive impairment, the overall  magnitude of the absolute values of the pairwise correlations are smaller, and interestingly there are fewer negative correlations. These effects are more pronounced at older age.  

\begin{figure}[!htb]
	\centering
	\includegraphics[width=.9\textwidth]{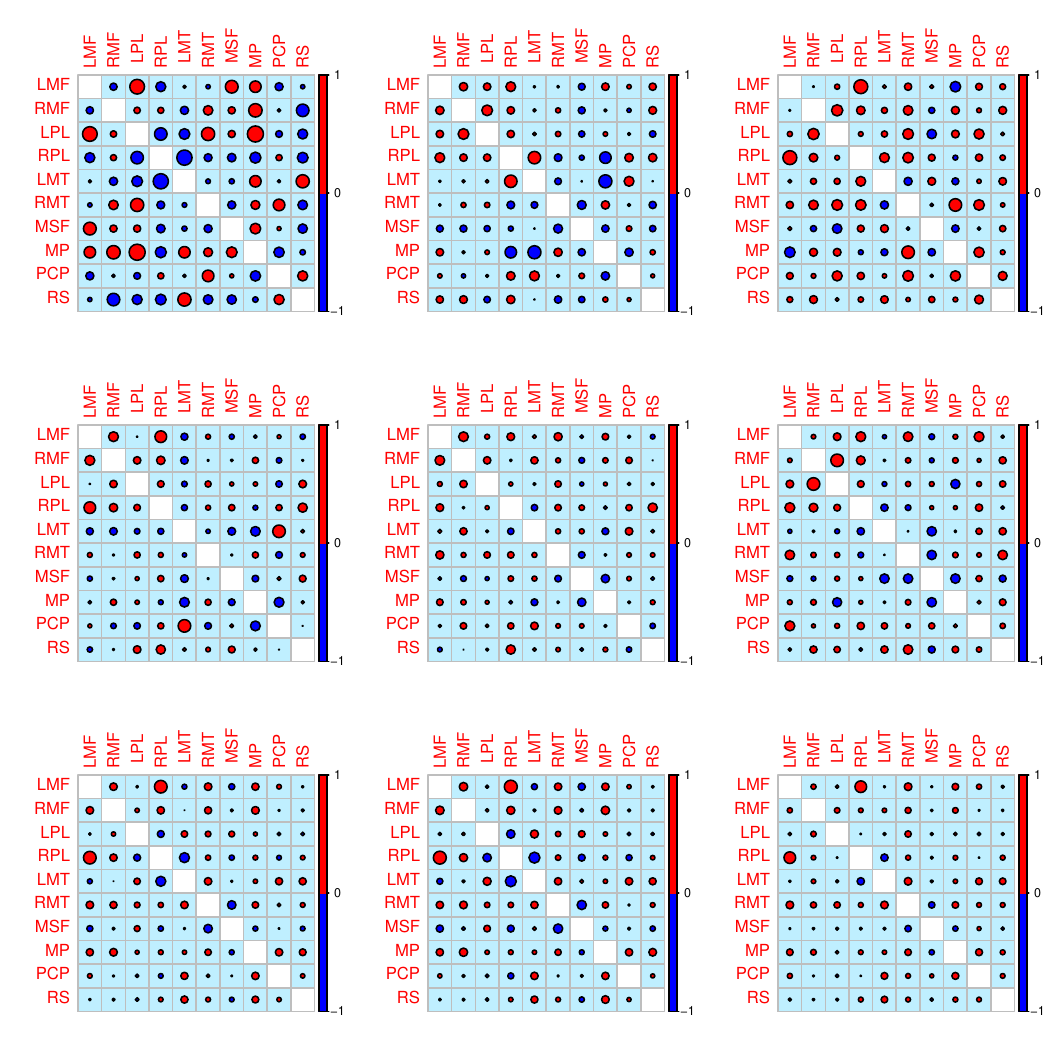}
	\caption{Estimated correlation matrix for the AD subjects fitted locally using nonparametric CORE in \eqref{eq:NPM:MODEL}. The top, middle and bottom rows show, respectively, the fitted correlation matrices at $10\%,\ 50\%,\ $ and $90\%$ quantiles of age. For each such age quantile, the columns (from left to right) depict the estimated correlation structure at $x = 10\%,\ 50\%,\ $ and $90\%$ quantiles of total cognitive score respectively. Positive (negative) values are drawn in red (blue) and larger circles correspond to larger absolute values. The figure illustrates the dependence of functional connectivity on total cognitive score, modulated by age.}
	\label{Figure AD_NPM_fit}
\end{figure}

Perhaps the most interesting finding from the fit (Figure~\ref{Figure AD_NPM_fit}) is the variation of  Negative Functional Connectivity (NFC) for the AD subjects \citep{zhou:10, brier:12, wang:07}. The positive pairwise correlations between the functional hubs, though reduced in magnitude, have a higher count  when moving  from a lower to a higher value in the total cognitive score direction. However, in the same context, the negative correlations diminish much more ostensibly in number and magnitude. Thus an increasing  reduction in the negative connectivity can be  associated with higher cognitive impairment,  and hence an increased cognitive impairment, in the AD subjects.

Also, the association between the functional connectivity and total cognitive score is modulated by age,  in the sense that at lower ages the association between cognitive impairment and reduction in Negative Functional Connectivity  is weaker than it is at higher ages.
Table~\ref{Tab:Corr:Magnitude:Diff} shows the difference, measured from the fits in Figure~\ref{Figure AD_NPM_fit}, between the total magnitude of the positive and the negative pairwise correlations, the latter being subtracted from the former. 
At each fixed  age $t$, the difference decreases with an increased value of the total cognitive score $x$, where the absolute values of the difference depend on age.  A similar concurrent or time-varying pattern in the estimated correlation matrices is also present for the CN subjects (Figure~\ref{A6:Fig:CN_NPM_fit} in the Supplementary Material Section~\ref{Appendix:A6}).


\begin{table}[]
	\centering
	\begin{tabular}{c|c|c|c}
		\hline
		& Lower Score & Median Score & Higher Score \\ \hline
		Lower Age  & 10.23       & 9.16         & 7.12         \\ \hline
		Median Age & 9.55        & 7.96         & 7.92         \\ \hline
		Higher Age & 9.27        & 7.70         & 7.25         \\ \hline
	\end{tabular}
	\caption{Difference in the total magnitude of positive correlations and the total magnitude of total negative correlation present in the estimated matrices in Figure~\ref{Figure AD_NPM_fit} at varying output points of total cognitive score and age. The lower, median, and higher levels are $10\%$, $50\%$, and $90\%$ quantiles, respectively, for both the total cognitive score and the age directions.}
	\label{Tab:Corr:Magnitude:Diff}
\end{table}

We also fitted the partially global CORE, as defined in \eqref{eq:PGM:MODEL}, to the same data and compared their performance, where the effect of total cognitive scores on the age-dependent functional connectivity correlation matrices is modeled as linear and the effect of age as nonparametric. To this end, the model was fitted separately for the AD and the CN subjects. The bandwidth parameter in the ``age'' direction was again chosen using a leave-one-out cross validation criterion and a Gaussian kernel was used. For the AD and CN subjects the optimal bandwidths were found to be $4.12$ and $3.22$, respectively.  We present the fits corresponding to the AD subjects over a range of output points in Figure ~\ref{Figure AD_PGM_fit}. We find a very similar pattern for the fitted correlation. The positive correlations increase in magnitude and quantity with increasing total cognitive score and age, while the curious changes in the Negative Functional Correlations are again noted.

\begin{figure}[!htb]
	\centering
	\includegraphics[width=.9\textwidth]{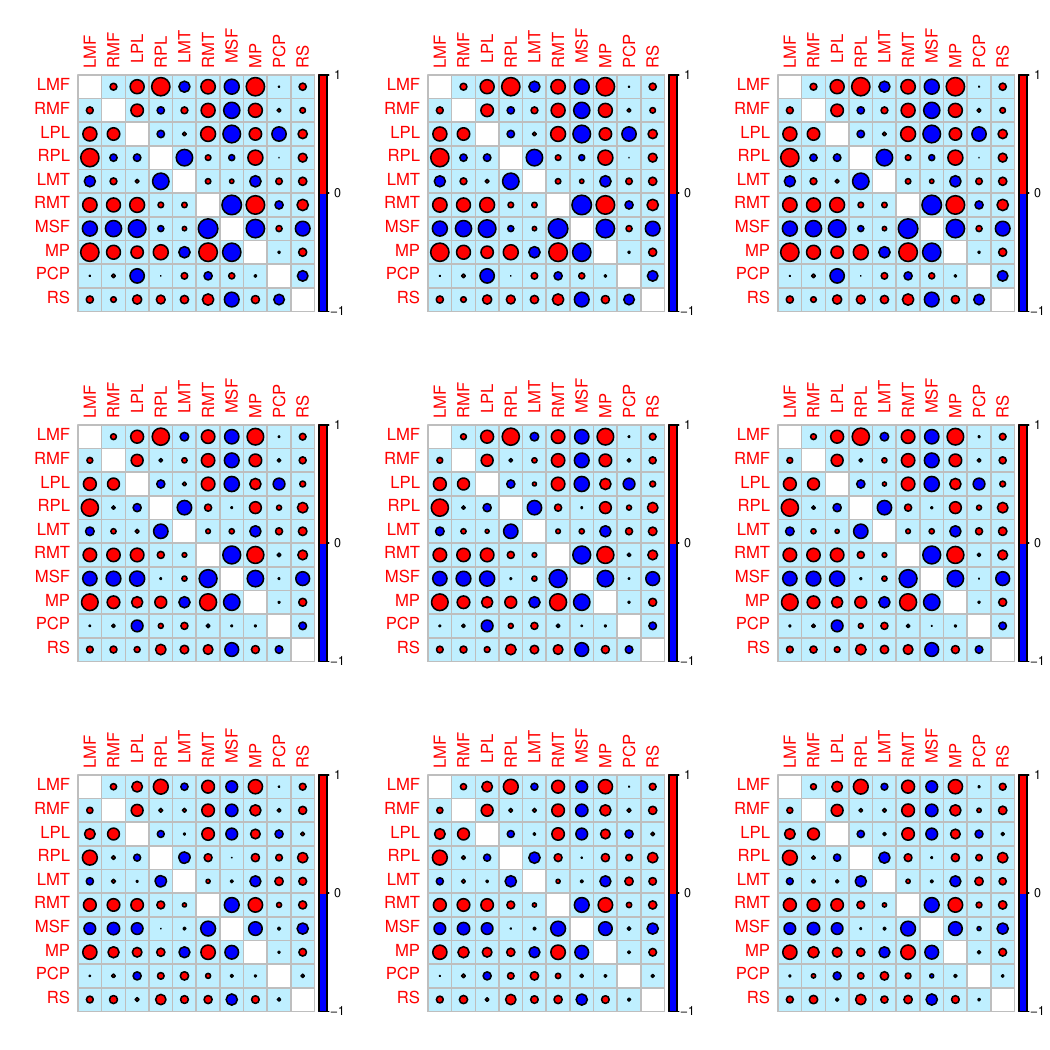}
	\caption{Estimated correlation matrix for the AD subjects fitted locally using partially global CORE in \eqref{eq:PGM:MODEL}. The top, middle and bottom rows show, respectively, the fitted correlation matrices at $10\%,\ 50\%,\ $ and $90\%$ quantiles of age. For each such age quantile, the columns (from left to right) depict the estimated correlation structure at $x = 10\%,\ 50\%,\ $ and $90\%$ quantiles of total cognitive score respectively. Positive (negative) values are drawn in red (blue) and larger circles correspond to larger absolute values. The figure illustrates the dependence of functional connectivity on total cognitive score, modulated by age.}
	\label{Figure AD_PGM_fit}
\end{figure}

To investigate the comparative goodness-of-fit of the two models, we  computed 
the average deviation of the fitted  from the observed correlation matrices over the age interval $[55,90]$,
\begin{align}
\text{MSE}_\oplus(t) := d_F^2(M_\oplus(t),\hat{M}_\oplus(t)),
\label{eq:FROB:DEV}
\end{align}
$M_\oplus(t)$ and $\hat{M}_\oplus(t)$ being the observed and fitted connectivity matrices, respectively,  at  age $t \in [55,90]$ and $d_F(\cdot,\cdot)$ the Frobenius distance between two correlation matrices.
Deviation \eqref{eq:FROB:DEV} is displayed in Figure~\ref{Fig:DistF:AD} for both the nonparametric and partially global CORE models. The partially global model seems to fit the data  better, which could indicate that the linear constraint for the impact of total cognitive score imposed in the partially Ggobal Core model is likely satisfied. 
The integrated deviance $\int_{\mathcal{T}} \text{MSE}_\oplus(t) dt$ is  $0.0570$ for the nonparametric CORE and $ 0.0494$ for the Partially CORE.

\begin{figure}[!htb]
	\centering
	\includegraphics[width=5cm]{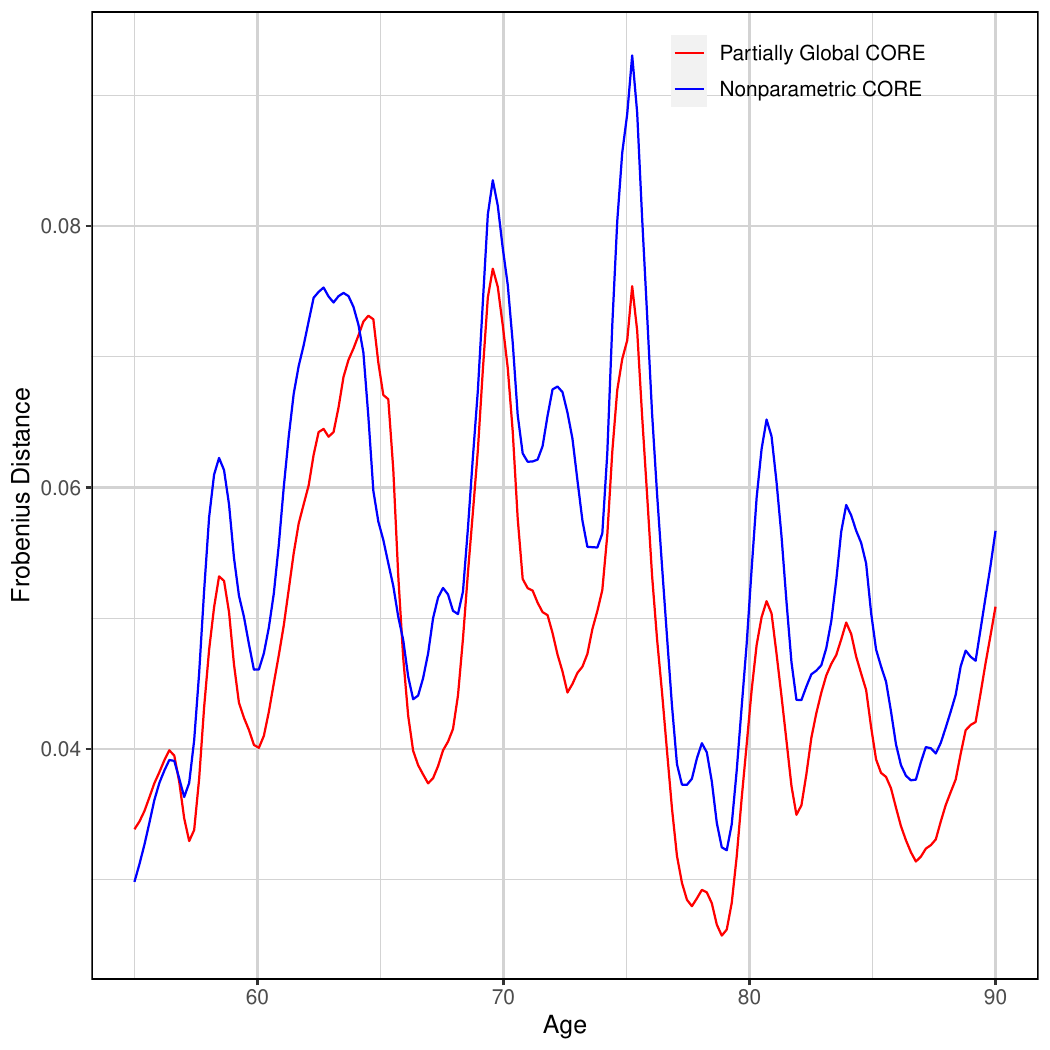}
	\caption{Comparison of fits for the two CORE models. The figure shows the Frobenius distances between the fitted and the observed correlation matrices across age for the AD subjects  using the nonparametric CORE model (blue) and the partially global CORE model (red) are illustrated.}
	\label{Fig:DistF:AD}
\end{figure}

We further look into the out-of-sample prediction performance of the two methods for the AD subjects and CN subjects separately. For this, we first randomly split the dataset into a training set with sample size $n_{\text{train}}$ and a test set with the remaining $n_{\text{test}}$ subjects. We then take the fitted objects obtained  from the training set, and predict the responses in the test set using the covariates present in the test set. As a measure of the efficacy of the fitted model, we  compute root mean squared prediction error as
\begin{align}
\text{RMPE} = \left[\frac{1}{n_{\text{test}}}\sum_{i=1}^{n_{\text{test}}} n_i^{-1}\sum_{l=1}^{n_i} d_F^2\left(Y^{\text{test}}_{il},\ \hat{l}_\oplus(X_{il},T_{il})\right)  \right]^{-1/2},
\end{align}
where $Y^{\text{test}}_{il}$ and $\hat{l}_\oplus(X_{il},T_{il})$ denote, respectively, the $i^{\text{th}}$ actual and predicted responses in the test set, evaluated at age $T_{il}$ and total cognitive score $X_{il}$. We repeat this process $1000$ times, and compute RMPE for each split for the AD and CN subjects separately (See Table~\ref{Tab:rmpe}).


\begin{table}[h!]
	\centering
	\begin{tabular}{  ccccc } 
		\hline
		&$n_{\text{train}}$ &  $n_{\text{test}}$ & nonparametric CORE& partially global CORE)\\ 
		\hline
		AD & $52$ & $26$ &  $0.306$ 
		&$0.322$ \\
		\hline
		CN & $271$ & $100$ &  $0.151$ 
		& $0.167$ \\
		\hline
	\end{tabular}
	\caption{Average Root Mean Prediction Error (RMPE) over $1000$ repetitions for the AD and CN subjects, as obtained from the local fits of the nonparametric and partially global  CORE models.  Here, $n_{\text{train}}$ and $n_{\text{test}}$ denote the sample sizes for the  split training and testing datasets respectively.}
	\label{Tab:rmpe}
\end{table}
We observe that the out-of-sample predictions errors are quite low for both the AD and CN subjects. In fact they are in the ballpark of the in-sample-prediction error, calculated as the average distance between the observed training sample and the predicted objects based on the covariates in the training sets, which supports the proposed CORE  models.  The nonparametric model shows a better predictive performance than the partially global CORE.  

To confirm the group differences in the time-varying structure of the correlation matrices we further conduct a permutation test. To test the null hypothesis that, for varying age and total cognitive score values, the AD and CN subjects have the same conditional correlation matrix objects, we use the heuristic test statistic, measuring the average discrepancy of the fit for the AD and CN groups as
\begin{align}
\int S(x,t)  \ dx \ dt =  \int d_F^2\left( \hat{\Sigma}^{\text{AD}}(x,t), \hat{\Sigma}^{\text{CN}}(x,t)\right) \ dx \ dt.
\end{align}
Here $\hat{\Sigma}^{\text{AD}}(x,t)$ and $\hat{\Sigma}^{\text{CN}}(x,t)$ denote the estimated correlation matrix objects at total cognitive score $x$ and age $t$, for the AD and CN subjects respectively, with $x \in [5,70]$ and varying age $t \in [55,90]$ and $d_F(\cdot,\cdot)$ is the Frobenius norm between two matrix objects.

All the observations are pooled, and the test statistic calculated for every possible way of dividing the pooled values into two groups of size $174$ and $694$. The set of these calculated test statistic values is the exact distribution of possible differences under the null hypothesis. The p-value of the test is calculated as the proportion of sampled permutations where the computed test-statistic value is more than or equal to the test statistic value obtained from the observed sample. Using $10^6$ permutation samples, and the estimation methods being the nonparametric CORE and partially global CORE, the p-values are found to be $0.009$ and $0.002$ , respectively. Thus both the methods are able to detect a significant difference in the functional connectivity between the AD and CN subjects, providing evidence that the CORE model is useful to differentiate these groups. 
A further look into the time-varying regression fits for those connectivity hubs that show a change in the magnitude of the correlations across the AD and CN subjects (Figure~\ref{A6:Fig:connectivity} in the Supplementary Material Section~\ref{Appendix:A6}) also indicates  differences between  the AD and CN subjects. 

\subsection{Impact of GDP on human mortality}
\label{data:GDP}
The Human Mortality Database (\url{https://www.mortality.org/}) provides yearly life table data differentiated by gender for $37$ countries across $50$ years. For our analysis, we considered the life tables for males according to yearly age-groups varying from age $0$ to $120$ for 22 countries over 14 calendar years, 1997-2010. Life tables can be viewed  as histograms, which then can be smoothed with local least squares to obtain smooth estimated probability density functions for age at death. We carried this out for each year and country, using the Hades package available at \url{https://stat.ucdavis.edu/hades/} for smoothing the histograms with a choice of the bandwidth as 2 to obtain the age-at-death densities. Thus these data can be viewed as a sample of 
time-varying univariate probability distributions, for
a sample of $22$ countries, where the time axis represents $14$ calendar years and the observations made at each calendar year for each country correspond to the age at death distribution, over the age interval $[0,120]$, for that year. 
An illustration of the time-varying age at death distributions represented as density functions over the calendar years for four selected countries is in  Figure \ref{Figure A1} in the Supplementary Material.

The data on GDP per capita at current prices is available at  the World Bank Database at \url{https://data.worldbank.org}. Considering the observed age-at-death densities for the countries over the calendar years as time-varying random objects that reside  in the space of distributions equipped with the Wasserstein-2 metric, and GDP per capita for these countries as real-valued time-varying covariates, we fit the proposed concurrent object regression (CORE) models as described  in Section \ref{NPM} and \ref{PGM}. Figure \ref{Figure 3} illustrates the time-varying nature of the fitted nonparametric CORE model, as per  \eqref{eq:NPM:MODEL}. We observe that for a fixed calendar year $t$ the fitted densities appear to shift towards the right as the value of the covariate GDP increases,  thus indicating that GDP per capita is positively associated with longevity at  a fixed calendar year.  If alternatively moving along the calendar years for a fixed GDP-value, one again observes an increasing trend in longevity. 

\begin{figure}
	\centering
	\begin{minipage}[t]{0.45\textwidth}
		\hspace*{1cm}	\includegraphics[width=.75\textwidth]{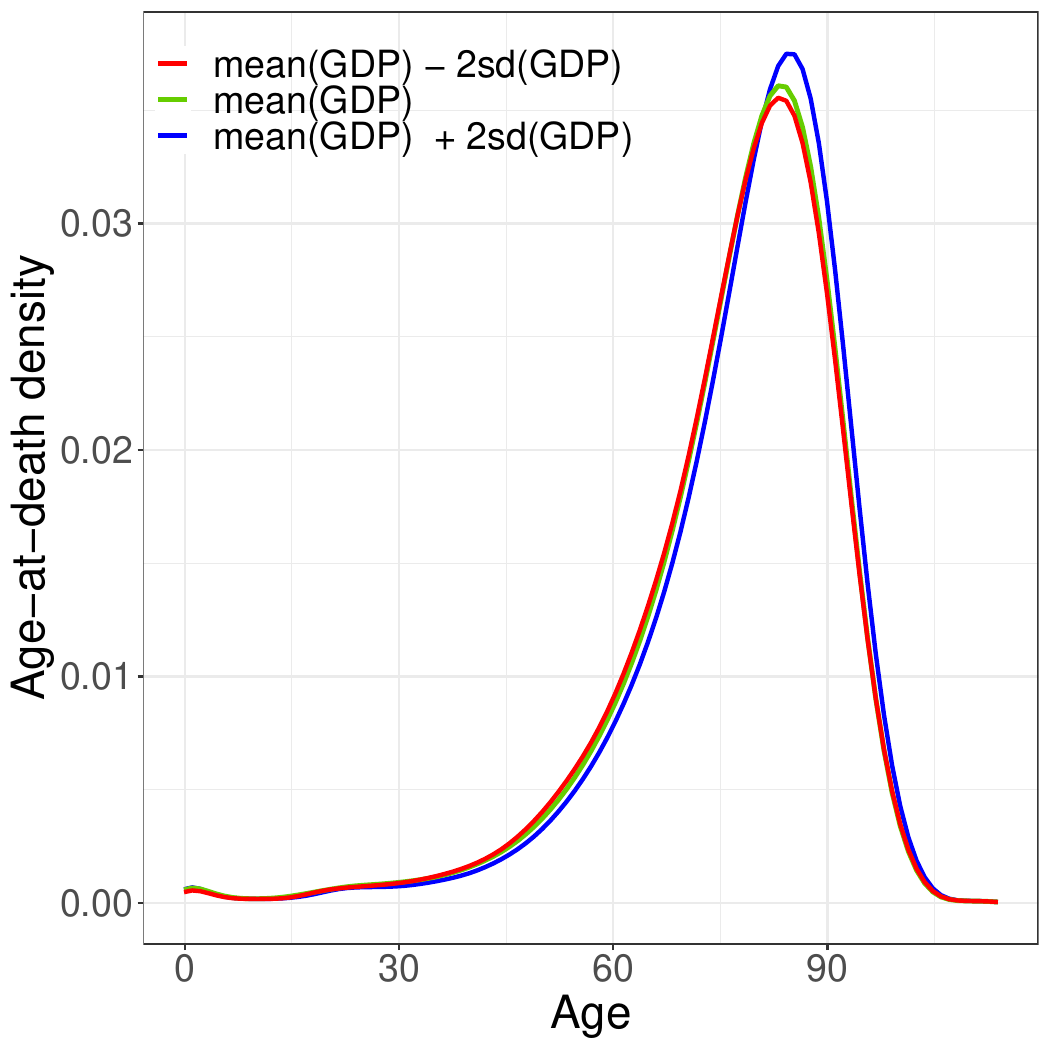}
		\label{Figure 3(a)}
	\end{minipage}
	\hfill 
	\begin{minipage}[t]{0.5\textwidth}
		\includegraphics[width = .9\textwidth]{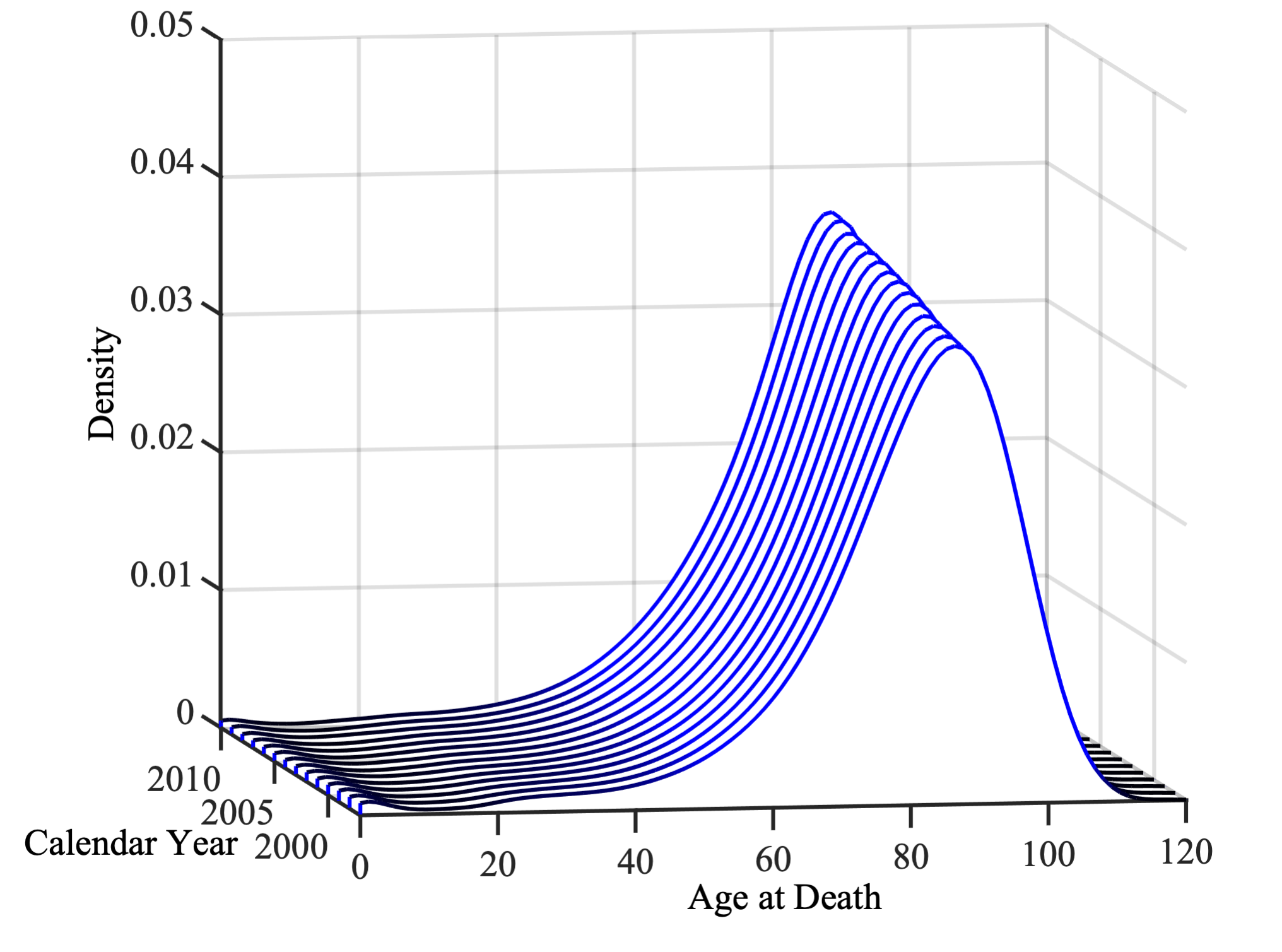}
		\label{Figure 3(b)}
	\end{minipage}
	\caption{Fitting the nonparametric concurrent object regression (CORE) model in \eqref{eq:NPM:MODEL}. In the left panel, the locally fitted densities of human mortality distributions, at the year $t = 2005$ and GDP value $x$ = mean(GDP)- 2 $\times$ sd(GDP), $x $= mean(GDP) and $x =$ mean(GDP) $ + 2 \times$ sd(GDP) are displayed in red, green and blue lines respectively.  The right panel shows the fitted densities for the U.S., varying over the years 1997-2010.}
	\label{Figure 3}
\end{figure}

Figure \ref{Figure 4} shows the 3D plots for the fitted densities over the years for four countries- Australia, Finland, Portugal and the U.S. We find that over the calendar years the modes for the age-at-death densities are shifted towards  older age and that the probability of death before age 5 declines for all the four countries, indicating increasing  life expectancy. Also, we notice that, for example, U.S. improves on child mortality over the years while for Finland it remains low throughout. These fits match quite well with  the observed densities in Figure \ref{Figure A1} (see the supplementary material).

\begin{figure}[!htb]
	\begin{minipage}{0.49\textwidth}
		\centering
		\includegraphics[width=.8\textwidth]{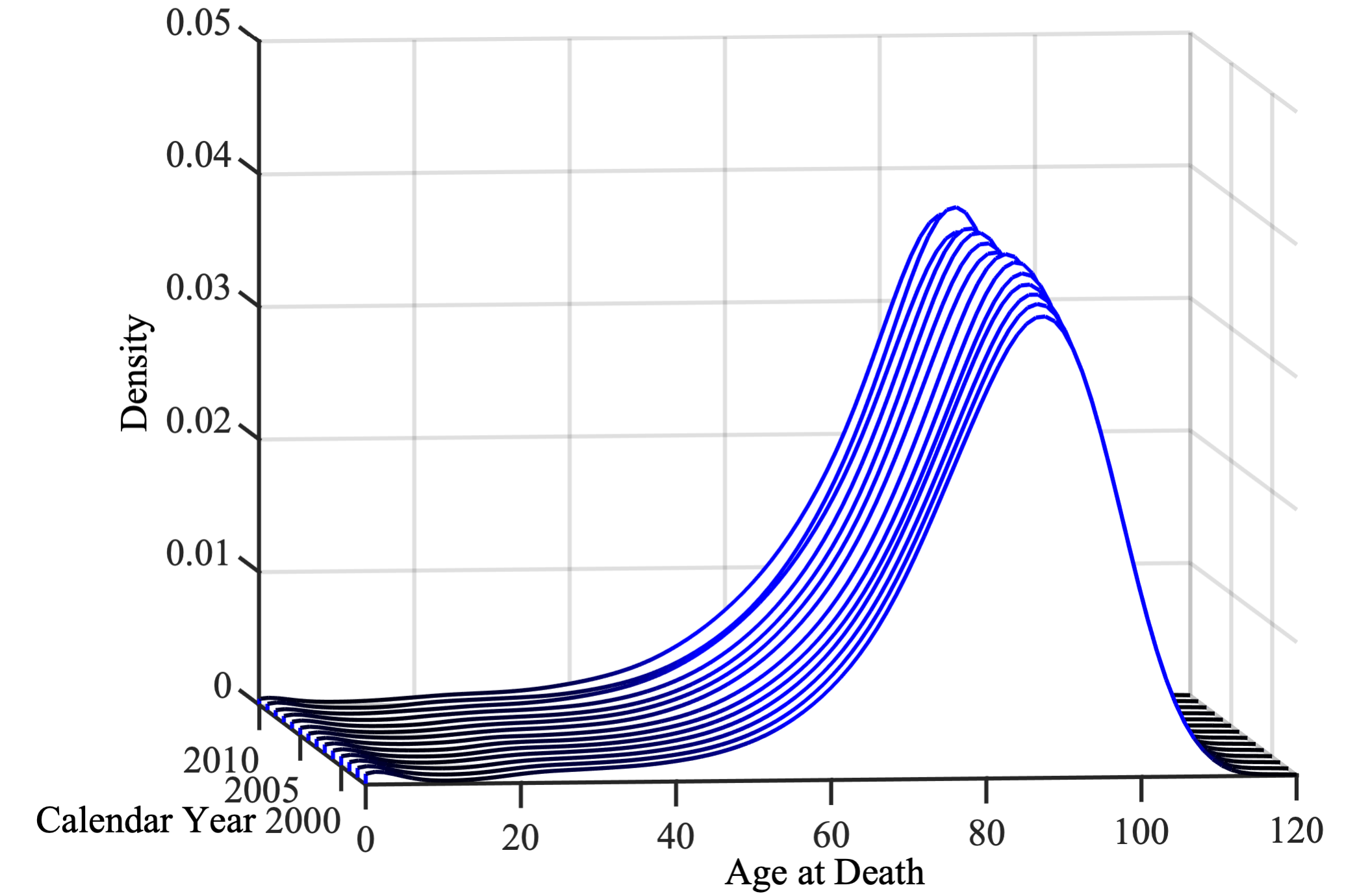}
	\end{minipage}\hfill
	\begin{minipage}{0.49\textwidth}
		\centering
		\hspace*{-.3cm}\includegraphics[width=.8\textwidth]{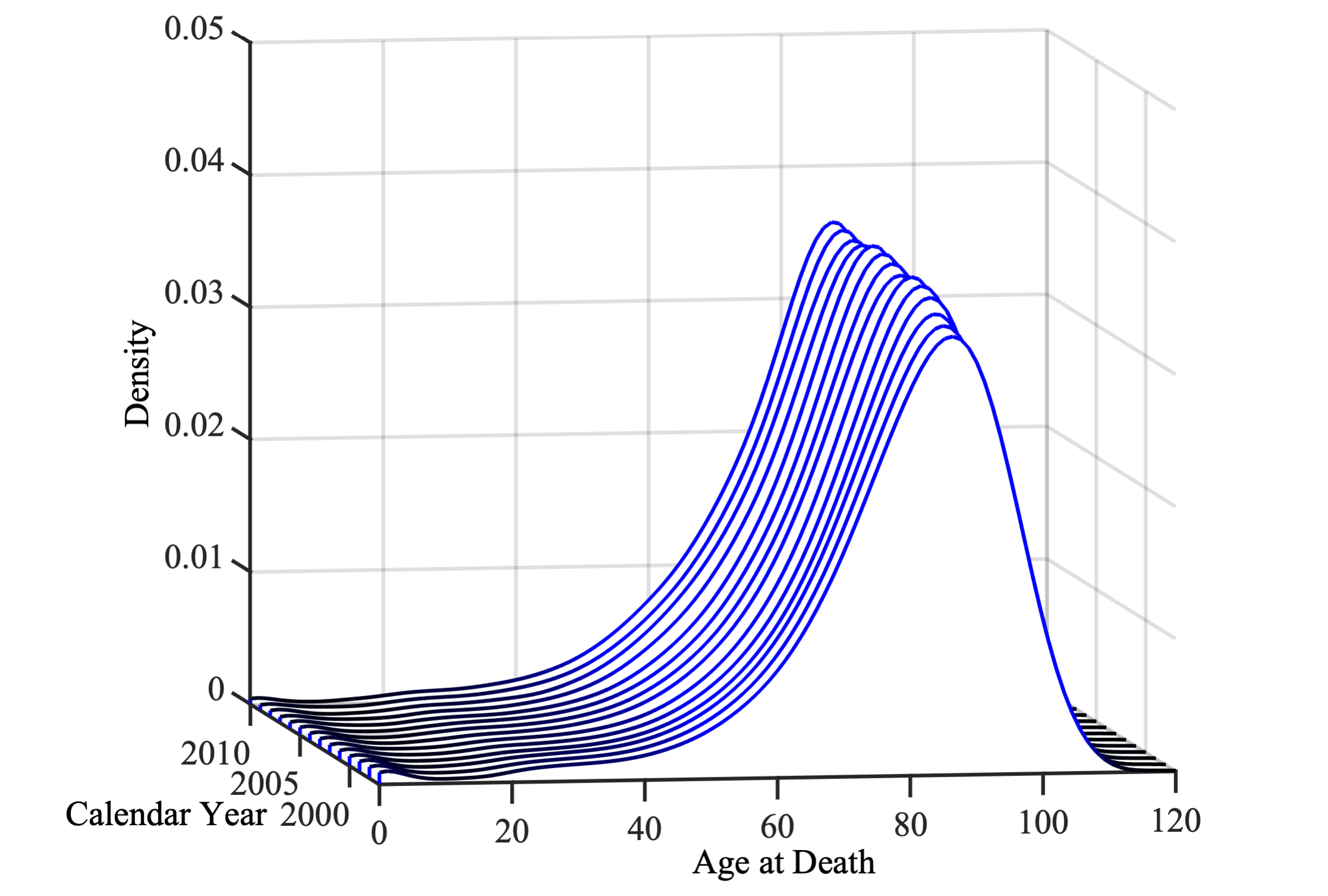}
	\end{minipage}\vspace*{.5cm}
	\begin{minipage}{0.49\textwidth}
		\centering
		\includegraphics[width=.8\textwidth]{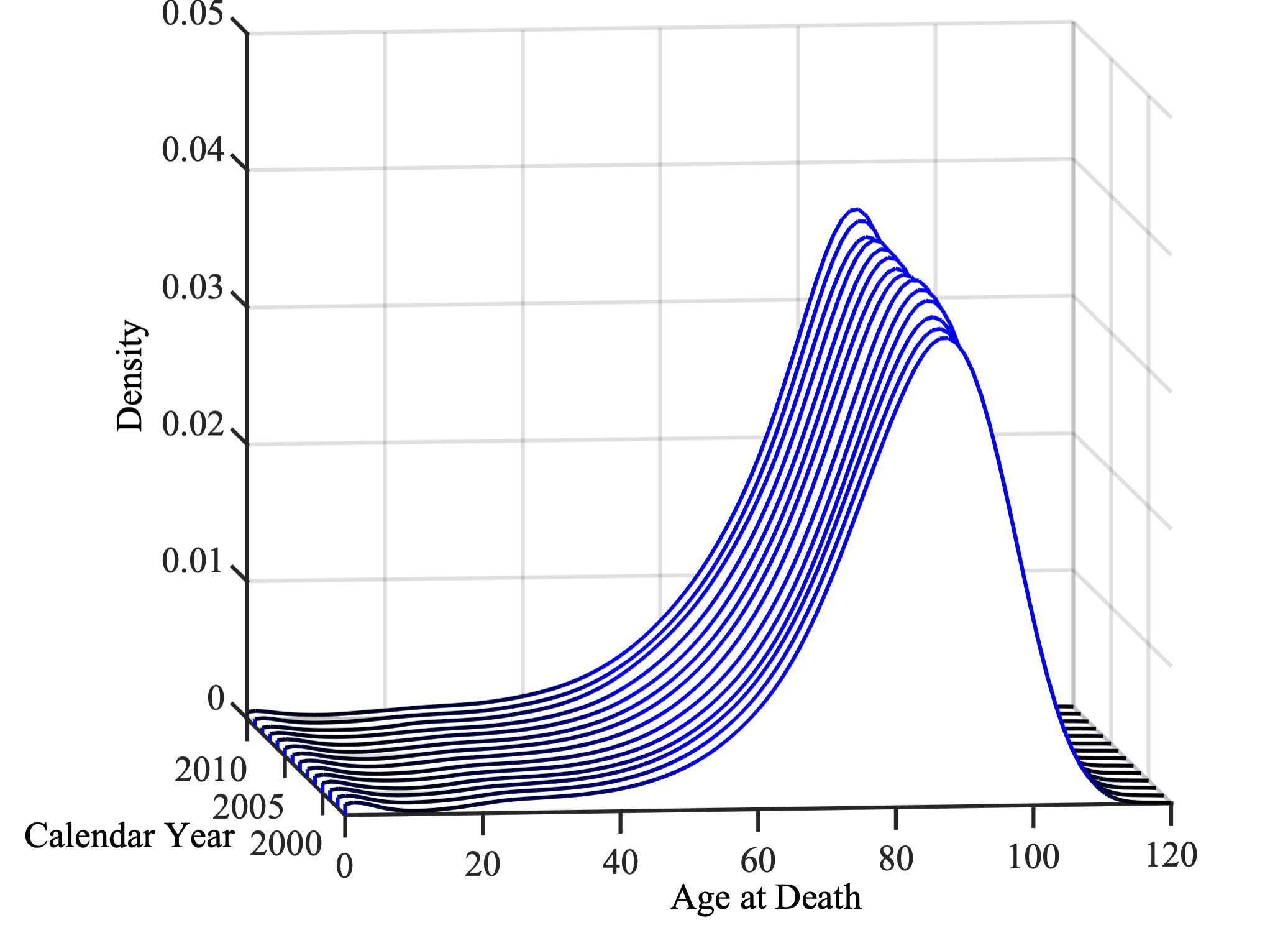}
	\end{minipage}
	\begin{minipage}{0.49\textwidth}
		\centering
		\includegraphics[width=.8\textwidth]{plots/us_pred-eps-converted-to}
	\end{minipage}
	\caption{Estimated age at death density functions over the years for males in Australia, Finland, U.S. and Portugal, clockwise in the four panels, starting at the upper left.}
	\label{Figure 4}
\end{figure}
We also fitted the partially global CORE, as defined in \eqref{eq:PGM:MODEL}, to the same data and compared their performance, where the effect of GDP is modeled as linear and the effect of calendar year as nonparametric. The left panel of Figure \ref{Figure 5} indicates that the fits are very similar at randomly chosen points $x = \text{mean(GDP)}; t =2005$.

\begin{figure}[!htb]
	\centering
	\hspace*{.65cm}\begin{minipage}[t]{0.48\textwidth}
		\includegraphics[width=.7\textwidth]{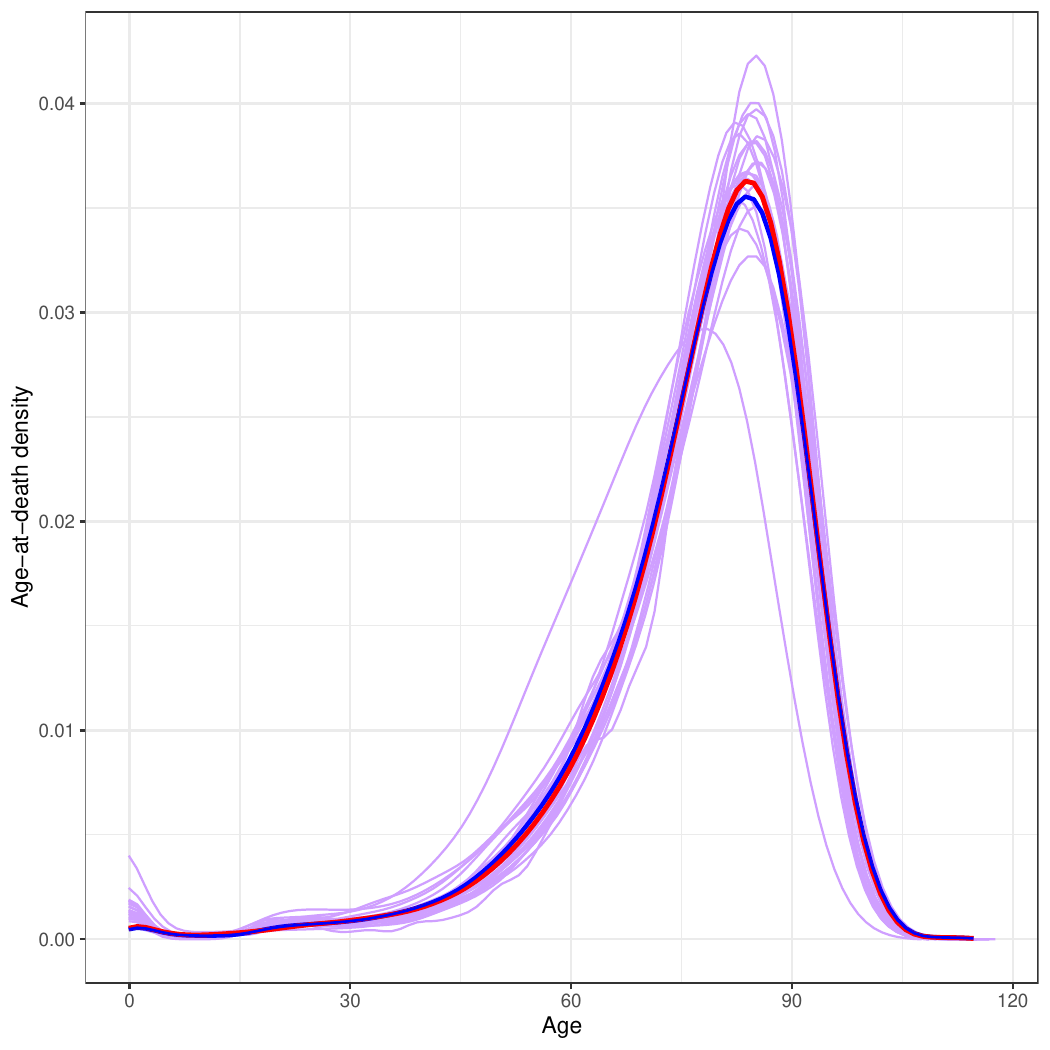}
		\label{Figure 5(a)}
	\end{minipage}
	\begin{minipage}[t]{0.47\textwidth}
		\includegraphics[width=.7\textwidth]{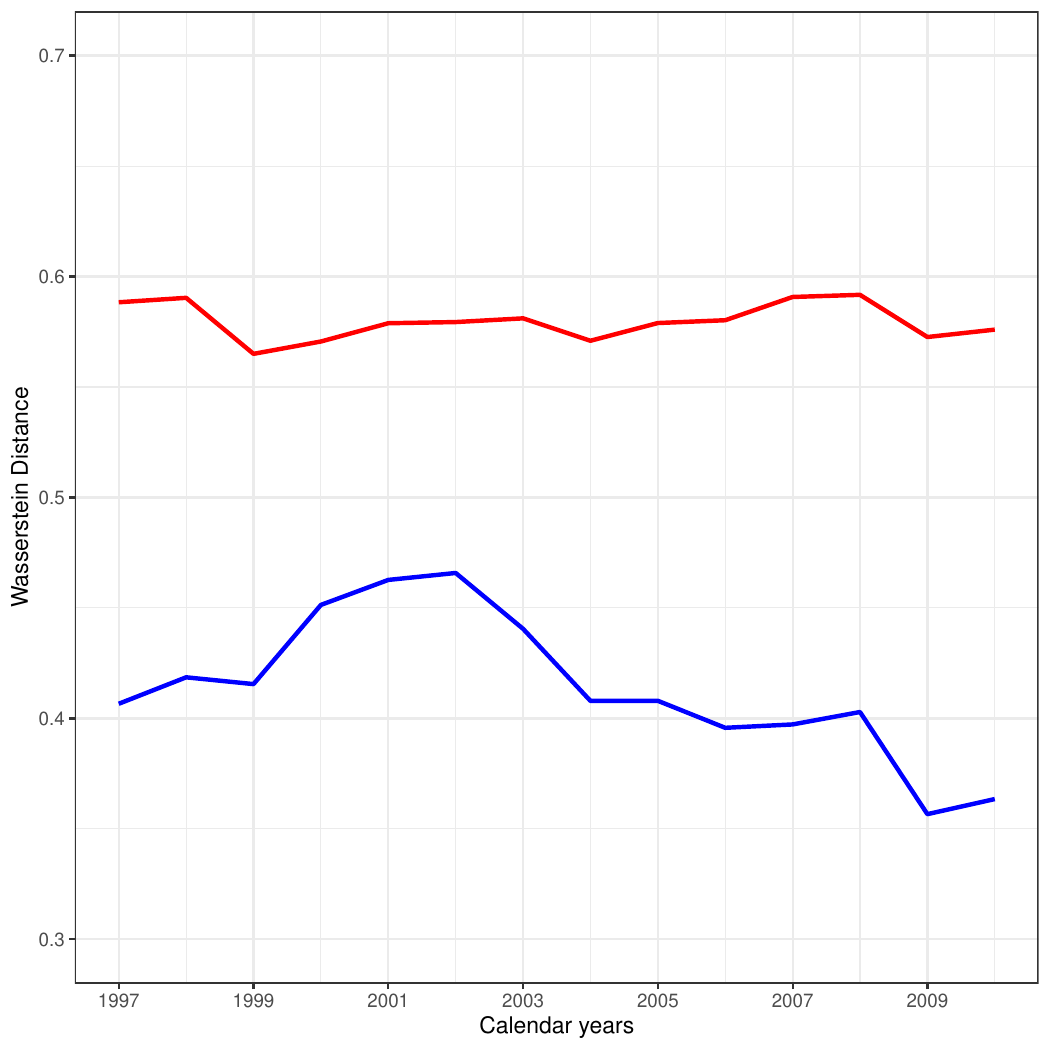}
		\label{Figure 5(b)}
	\end{minipage}
	\centering
	\caption{Comparing the fits  of partially global \eqref{eq:PGM:MODEL} and nonparametric \eqref{eq:NPM:MODEL} concurrent object regression (CORE) models. The left panel shows the local fits at the points $x = \text{median}(X), \,  t = 2005$, comparing both models. The blue and  red curves represent the nonparametric and the partially global regression fits, respectively. The purple curves in are the observed densities for the year $2005.$ In the right panel the average Wasserstein distances between the fitted and the observed densities across the calendar years for the nonparametric model (blue) and  the partially global model (red) are illustrated.}
	\label{Figure 5}
\end{figure}
For both  models, the bandwidth $h$ is chosen by leave-one-out cross validation method, as the minimizer of the mean discrepancy between the regression estimates and the observed age-at-death density functions and a Gaussian kernel is used. To investigate the comparative goodness-of-fit of the two models further, we  computed 
the average deviation of the fitted  from the observed densities for each of the 14 calendar years as
\begin{align}
\text{MSE}_\oplus(t) := d_W(f_\oplus(t),\hat{f}_\oplus(t)),
\label{eq:WASS:DEV}
\end{align}
$f_\oplus(t)$ and $\hat{f}_\oplus(t))$ being the observed and fitted age-at-death densities, respectively,  at  calendar years $t \in \{1997,\dots 2010\}$ and $d_W(\cdot,\cdot)$ the Wasserstein-2 distance between two densities (distributions).
Deviation \eqref{eq:WASS:DEV} is displayed in the right panel of Figure~\ref{Figure 5} for both the nonparametric and partially global CORE models. The nonparametric model seems to fit the data  better, which could indicate that the linear constraint for the impact of GDP imposed in the partially global Core model is likely not satisfied.   
The integrated deviance $\int_{\mathcal{T}} \text{MSE}_\oplus(t) dt$ is  $0.413$ for the nonparametric CORE and $0.580$ for the Partially CORE. 

\section{Concluding remarks}\label{CONCL}
The proposed concurrent object regression (CORE) is useful for  the regression analysis of random objects, where it complements Fr\'echet regression, by extending the notion of conditional Fr\'echet means further to a  concurrent or varying coefficient framework. We provide theoretical justifications including rates of pointwise convergence  for both global and local versions of the CORE model, and a uniform convergence result for the global part.  
For the special case of Euclidean objects the rates of convergence correspond to the known optimal rates.  The rate of convergence for the nonparametric CORE model is intrinsically connected to an inherent manifold structure of the predictor space. 
Analogously to local regression, the nonparametric estimators will suffer from the curse of dimensionality if the predictor space is of  higher  dimension than $p=2$ or $p= 3$.  This calls for future research in dimension reduction in the predictor space. 
A feature of interest is that we  do not require observing the  complete stochastic processes $\{\left(X(t), Y(t)\right) : t\in \T \}$ but only need samples taken at random predictor times, and our methods can be adapted for sparse and longitudinal predictors.

\section{Acknowledgments}
	Data collection and sharing for this project was funded by the Alzheimer's Disease
	Neuroimaging Initiative (ADNI) (National Institutes of Health Grant U01 AG024904) and DOD ADNI (Department of Defense award number W81XWH-12-2-0012). As such, the investigators within the ADNI contributed to the design and implementation of ADNI and/or provided data but did not participate in analysis or writing of this report. A complete listing of ADNI investigators can be found at: \url{http://adni.loni.usc.edu/wp-content/uploads/how_to_apply/ADNI_Acknowledgement_List.pdf}.

\section{Appendix}
\subsection{Background for the partially global  concurrent object regression}\label{Appendix:A2}

Motivation of deriving  \eqref{eq:PGM:MODEL} for the Euclidean response case. 	When $(\Omega,d) = (\reals,d_E)$, we write $m_\oplus(\cdot,\cdot) = m (\cdot,\cdot)$. Assuming the true relation between the response $Y$ and the predictor $X(T)$ is  linear while there is a smooth  nonparametric relation in the $T$ direction, a partially local linear type estimator of the regression model $m(\cdot,\cdot)$ at the point $T =t,\ X(T) =x$ is given by
		$
		\hat{m}(x,t) = \hat{a}^T(x-\mu_X(t)) +\hat{\beta}_0,
		$
		\text{where} $\mu_X(t) =\E{X|T=t} =\E[X|T=t]{X(t)}\text{  for all } t \in \T.$ This can be written alternatively as
		\[
		(\hat{a}, \hat{\beta_0},\hat{\beta_1}) = \underset{a, \beta_0,\beta_1}{\text{argmin}} \frac{1}{n} \sum_{i=1}^{n} \left[ \frac{1}{n_i} \sum_{l=1}^{n_i} K_h(T_{il} - t)(Y_{il} - a^T(X_{il}-\mu_X(t))- \beta_0 - \beta_1(T_{il} - t))^2\right].
		\]
		We can view this as an M-estimator of an intermediate population model,
		\[\tilde{g}(x,t) = (a_1^*(x,t))^\intercal(x-\mu_X(t)) + \beta_0^*(t), \text{ where} \]
		
		\begin{align*} &(a_1^*, \beta_0^*,\beta_1^*)\\ &=\underset{a_1, \beta_0,\beta_1}{\text{argmin}} \displaystyle \int \left[ \int y dF_{Y|X,T}(y,x,t) - a_1^\intercal(x-\mu_X(t)) - \beta_0 - \beta_1(s-t)\right]^2 K_h(s-t) dF_{X,T}(x,s).
		\end{align*}
\noindent Defining as before the following auxiliary parameters for $j=0,1,2, $\\\\
		$
		\mu_{0j} :=  \E{K_h(T-t)(T-t)^j} , \quad  \Sigma_{2j} := \E{K_h(T-t)(T-t)^j(X(T)-\mu_X(t))(X(T)-\mu_X(t))^\intercal}  ,\label{eq:a2.5}\\ 
		r_{0j} := \E{K_h(T-t)(T-t)^jY} , \quad r_{1j} := \E{K_h(T-t)(T-t)^j Y(X(T)-\mu_X(t))} , \nonumber \\ \sigma_0^2 := \mu_{02}\mu_{00} - \mu_{01}^2. \label{a2.6}
		$\\\\
		and solving the minimization problem leads to \\
		\[a_1^* = \Sigma^{-1}_{20} r_{10}, \quad \beta_0^* = \frac{r_{00}\mu_{02}-r_{01}\mu_{01}}{\sigma_0^2},\quad \beta_1^* = \frac{r_{01}\mu_{00}-r_{00}\mu_{01}}{\sigma_0^2}.\]
		Putting the optimal values of the parameters back in the model,\\
		\[
		\tilde{g} (x,t )  = a_1^*(x,t)(x-\mu_X(t)) + \beta_0^*(x,t)  = \displaystyle \int s^G(z,x,s,t,h) y dF(y,z,s) 
		= \E{s^G\left(X,x,T,t,h\right)Y},\label{eq:a2.7}
		\]
		with  weight function,
		\[	s^G(z,x,s,t,h) = \underbrace{K_h(s-t)\left[(z-\mu_X(t))^\intercal  \Sigma^{-1}_{20} (x-\mu_X(t))\right]}_{:=s_1(z,x,s,t,h)} + \\
		\underbrace{\frac{1}{\sigma_0^2} K_h(s-t)\left(\mu_{02}-(s-t)\mu_{01} \right)} _{:=s_2(s,t,h)} \label{a2.8}.
		\]
		Rewriting the framework as the weighted Fr\'echet mean w.r.t the Euclidean metric,
		\[
		\tilde{g}(x,t) = \underset{y \in \reals}{\text{argmin}}\  \E{s^G(X,x,T,t,h)(Y-y)^2} = \underset{y \in \reals}{\text{argmin}}\ \E{s^G(X,x,T,t,h)d_E^2(Y,y)},
		\]
		where $\tilde{g}$ can be viewed as a smoothed version of the true regression function $m$ with bias $m(x,t) - \tilde{g}(x,t) = o(1)$. This alternative formulation of the combination of a global and a local regression component thus provides the  intuition to define the general population model for metric-space valued random objects as
		\[
		\tilde{g}_\oplus(x,t) = \underset{\omega \in \Omega}{\text{argmin }} \tilde{G}_\oplus(\omega), 
		\text{where, } \ \tilde{G}_\oplus(\omega) := \E{s^G(X,x,T,t,h)d^2(Y,\omega)}.
		\]
\subsection{Technical assumptions (B1)-(B6) and (U1)-(U4) in section~\ref{PGM}}\label{Appendix:A3}
	
	
	The following is a list of these assumptions which are required for section 4. 
	\begin{enumerate}[label = (B\arabic*)]
		\item \label{B1} The kernel function $K$ is a univariate probability density that is symmetric around zero, with $|K^\gamma_{0j}|=| \int K^\gamma(u)u^j \ du \displaystyle| < \infty\text{ for } j=1,\dots, 4$ and $\gamma = 0,1,2.$
		\item \label{B2} The marginal density $f_{(X,T)}(x,t)$  and the conditional density $f_{(X,T)|Y}(x,t,y)$ exist, are twice continuously differentiable as a function of $t$ for all $x$ and  all $y$. 
		\item \label{B3} The Fr\'echet means $m_\oplus(x,t), \tilde{g}_\oplus(x,t), \hat{g}_\oplus(x,t)$ exist and are unique.
		\item \label{B4} For any $\epsilon>0,$
		\[
		\underset{d(\omega,m_\oplus(x,t))>\epsilon}{\inf} \left(M_\oplus(\omega,x,t)- M_\oplus(m_\oplus(x,t),x,t)\right) >0.\]
		\[\underset{d(\omega,\tilde{g}_\oplus(x,t))>\epsilon}{\inf} \left(\tilde{G}_\oplus(\omega,x,t)- \tilde{G}_\oplus(\tilde{g}_\oplus(x,t),x,t)\right) >0.\]
		\item \label{B5} There exist $\eta_1 >0,\,  C_1 >0$, \, $\text{with } d(\omega,m_\oplus(x,t)) < \eta_1$ such that \\
		\[
		M_\oplus(\omega,x,t) - M_\oplus(m_\oplus(x,t),x,t) \geq C_1d(\omega,m_\oplus(x,t))^2.
		\]
		\item  \label{B6} There exist $\eta_2 >0, C_2 >0$, $\text{with } d(\omega,\tilde{g}_\oplus(x,t)) < \eta_2$ such that \\
		\[
		\underset{N}{\text{lim inf }} \left[ \tilde{G}_\oplus(\omega,x,t)  - \tilde{G}_\oplus(\tilde{g}_\oplus(x,t),x,t) \right]  \geq C_1d(\omega,\tilde{g}_\oplus(x,t))^2.
		\]
	\end{enumerate} 
	These assumptions are required to ensure the existence and uniqueness of the Fr\'echet mean in the population and sample cases and the local curvature of the objective functions near their respective minimums to establish consistency of the partially global concurrent object regression (CORE) estimator. Also the relevant entropy conditions are necessary to prove the rate of convergence of the CORE estimator.\\
	For proving the uniform convergence results in the $X$-direction for any fixed value of $t$, the following additional conditions are used.
	\begin{enumerate}[label = (U\arabic*)]
		\item \label{U1} For almost all $x$ such that $||x||_E \leq M $, the Fr\'echet means $m_\oplus(x,t), \tilde{g}_\oplus(x,t), \hat{g}_\oplus(x,t)$ exist and are unique.
		\item \label{U2} For any $\epsilon>0$, \[\underset{||x||_E \leq M}{\inf} \underset{d(\omega,m_\oplus(x,t))>\epsilon}{\inf} \left(M_\oplus(\omega,x,t)- M_\oplus(m_\oplus(x,t),x,t)\right) >0.\]
		Also, there exists $\zeta = \zeta(\epsilon)$ such that
		\[
		P\left(\underset{||x||_E \leq M}{\inf} \underset{d(\omega,\hat{g}_\oplus(x,t))>\epsilon}{\inf} \hat{G}_\oplus(\omega,x,t)- \hat{G}_\oplus(\hat{g}_\oplus(x,t),x,t ) \geq \zeta \right) \to 1.\]
		\item \label{U3} With $\mathcal{B}_\delta(m_\oplus(x,t))$ and $N(\epsilon, \mathcal{B}_\delta(m_\oplus(x,t)),d)$, as defined in Assumption~(A\ref{A7})
		\[\int_{0}^{1} \underset{||x||_E \leq M}{\sup}\sqrt{1+ \log N(\delta\epsilon, \mathcal{B}_\delta(m_\oplus(x,t)),d)} d\epsilon = O(1) \text{ as} \delta \rightarrow 0.\]
		\item \label{U4} There exist constants $\tau>0, D>0$ and $\alpha>2$ possibly depending on $M$ such that, for any given $t$,\\
		\[\underset{||x||_E \leq M}{\inf} \underset{d(\omega, m_\oplus(x,t))<\tau}{\inf} M_\oplus(\omega,x,t) - M_\oplus(m_\oplus(x,t),x,t) - Dd(\omega,m_\oplus(x,t)) \geq 0.\]
	\end{enumerate}

	\subsection{Additional figures} \label{Appendix:A6}
	We present here some additional figures that are referred to in the main paper in the context of simulation studies and real data applications in Sections~\ref{SIM} and~\ref{DATA} respectively.
	
	\subsubsection*{Additional figure from simulation studies in Section~\ref{SIM}}
	
	The performance of the proposed partially global concurrent object regression (CORE) model is compared to the global Fr\'echet regression (GFR) method from~\cite{pete:mull:19}. In the latter, the nested structure of the predictor space $(T,X(T))$ is ignored and thus $T\in \reals$ and $X\in\reals^{p}$ are treated as a $p+1$ dimensional predictor input for the model. The data generating mechanism is as described in Setting I of Section~\ref{SIM:dens}, with the each component of the predictor process $X(\cdot) \in \reals^p$ assumed to be uncorrelated. The proposed partially global CORE method outperforms GFR in all cases. 
	\begin{figure}[!htb]
		\centering
		\includegraphics[width=\textwidth]{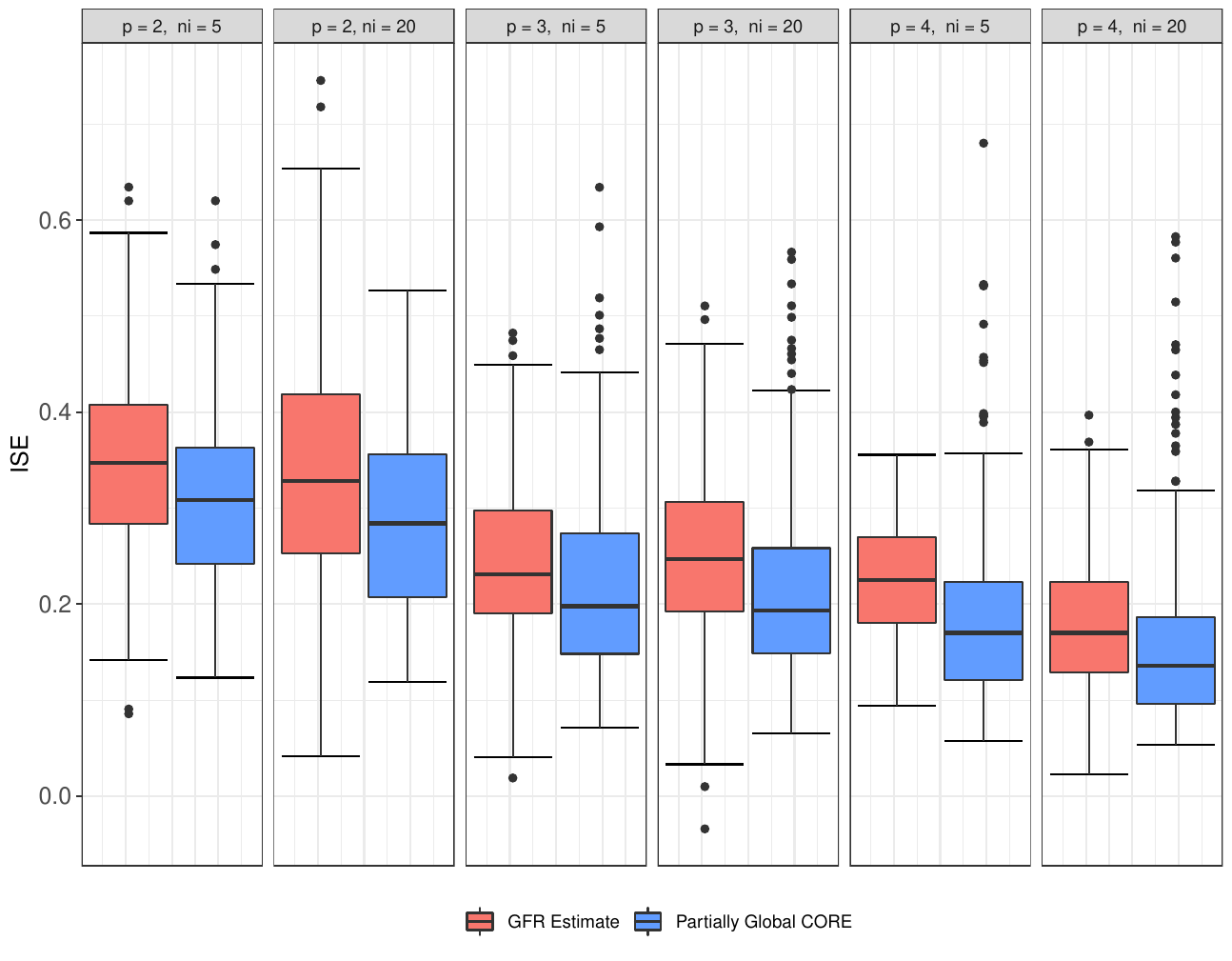}
		\caption{Figure shows the comparative performance of the proposed partially global concurrent object regression (CORE) method to that of global Fr\'echet regression (GFR) with increasing the predictor dimension $p$ for distributional object responses. The sample sizes are kept fixed at $n= 1000$ and dense and sparse designs are considered with $n_i =5 $ and $n_i =20$ respectively. }
		\label{sim:dens:dim}
	\end{figure}
	\clearpage
	
	\subsubsection*{Additional figures from real-data applications in Section~\ref{DATA}}
	The following figures show additional illustrations for the data application for brain connectivity in Alzheimer's disease in Section~\ref{DATA:fMRI}, where pairwise connectivity correlation matrices are considered as random object responses varying with age, and the predictors taken were age and cognitive score changing with age.
	\begin{figure}[!htb]
		\centering
		\includegraphics[width=.8\textwidth]{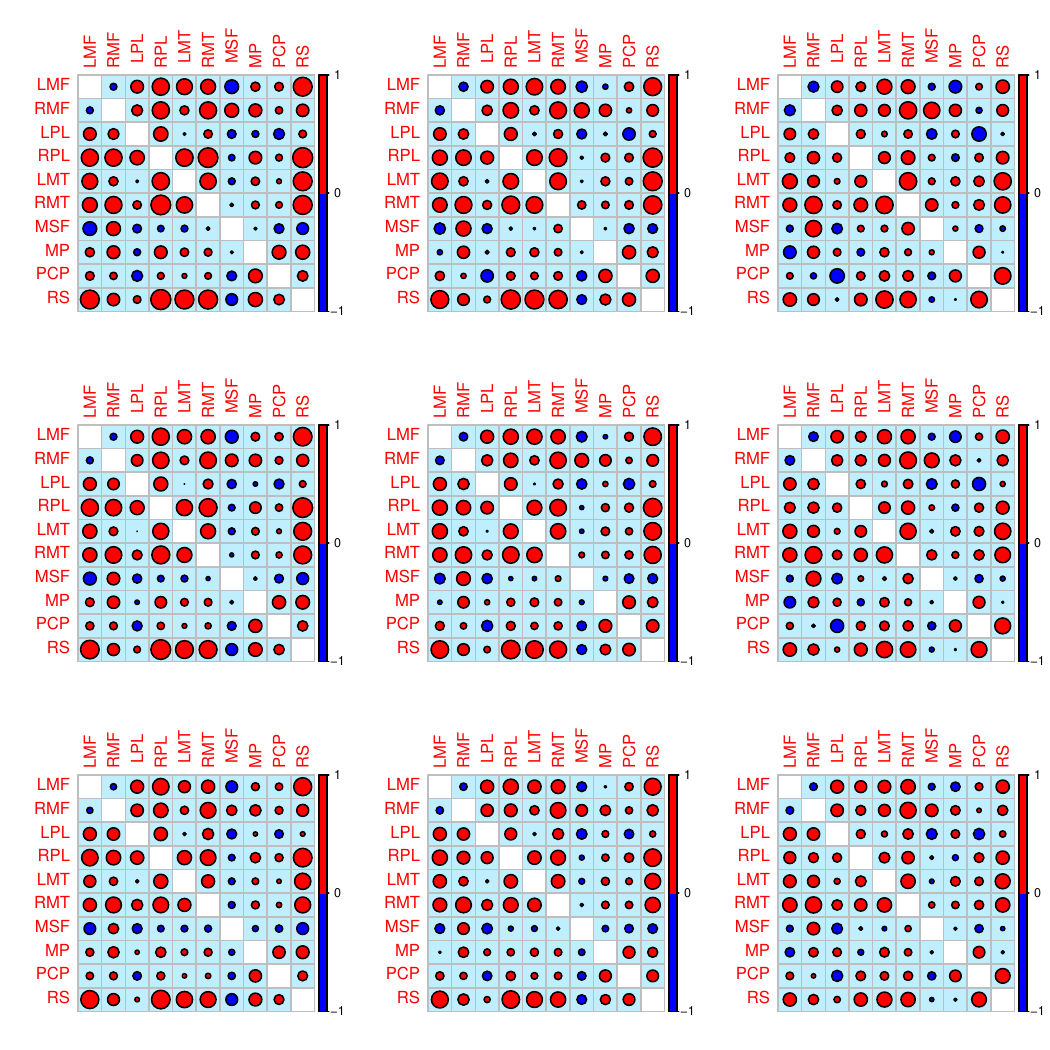}
		\caption{Estimated correlation matrix for the CN subjects fitted locally using nonparametric CORE in \eqref{eq:NPM:MODEL} illustrating the dependence of functional connectivity on total cognitive score which gets modulated by age. The arrangement of the panels are the same as that of Figure~\ref{Figure AD_NPM_fit}.}
		\label{A6:Fig:CN_NPM_fit}
	\end{figure}
	
	Figure~\ref{A6:Fig:CN_NPM_fit} displays the connectivity correlation matrices for the CN subjects, estimated using the nonparametric CORE method locally over a score of different output points. This elicits a the regression relationship between the functional connectivity matrix and the total cognitive scores in Section~\ref{DATA:fMRI}, which is further altered by age. Quite contrary to the case of the AD subjects (~\ref{Figure AD_NPM_fit}), here we observe a prominence of positive correlations between the brain parcellation throughout, in terms of stronger magnitude and higher number. This might well be indicatve of a better inter-hub functional connectivity in the CN subjects. Over increasing age we observe a higher value for the total cognitive score which can be associated with a weaker inter-hub connectivity overall. 
	The reduction in Negative Functional Correlation (NFC) for CN subjects is still noted but the evolution is not so drastic over age.
	In addition, the estimated correlation matrices for the CN subjects exhibit specific patterns of dependency over the connectivity hubs, which, in case of the estimated correlation matrices for the AD subjects is not as discernible. A further application of the partially global model gives evidence along the same line as the nonparametric CORE model (Figure~\ref{A6:Fig:CN_PGM_fit}). However, the in-sample goodness of fit measured by the integrated deviance statistic (see \eqref{eq:FROB:DEV} in Section~\ref{DATA:fMRI}) for the former ($0.0056$) is marginally better than the latter ($0.0071$), accounting for a better performance of the partially global Model.

	\begin{figure}[!htb]
		\centering
		\includegraphics[width=.8\textwidth]{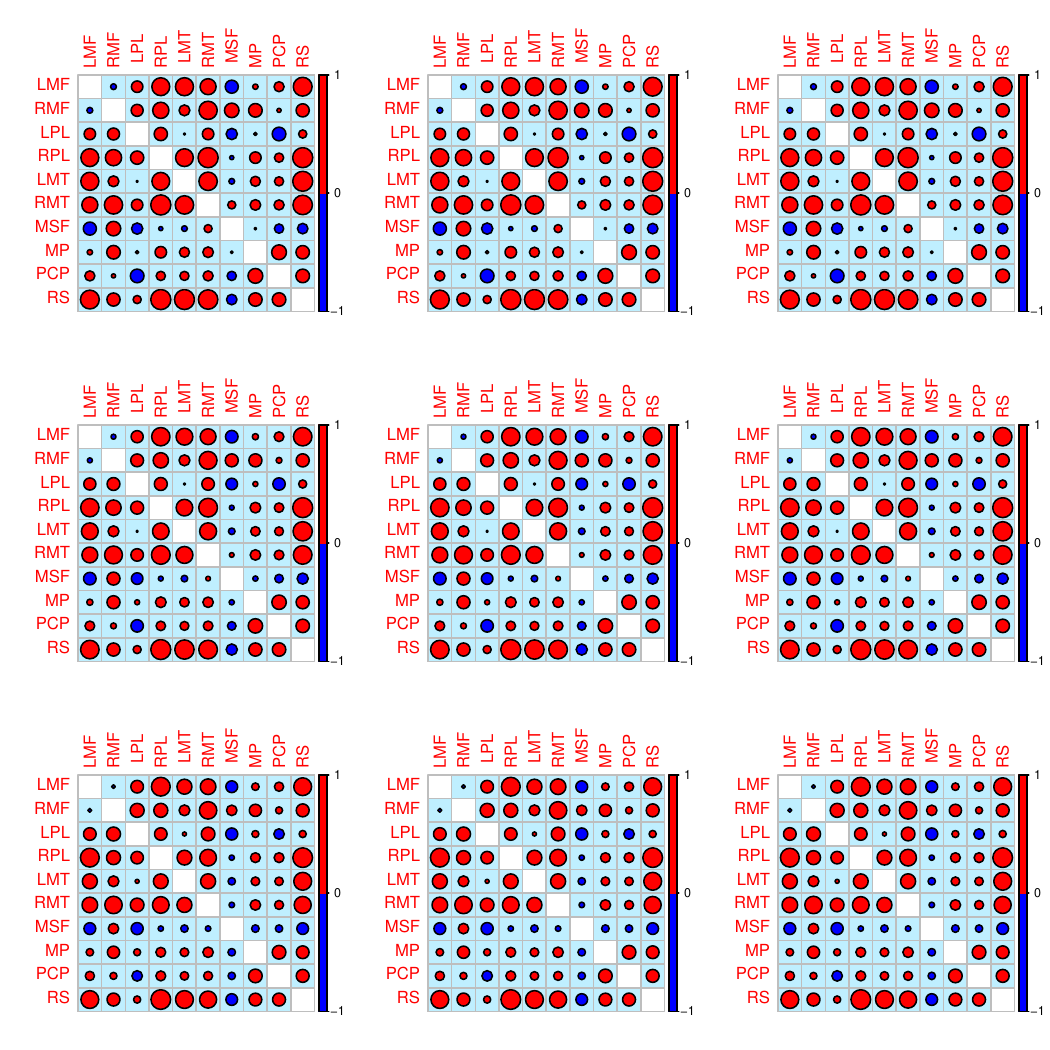}
		\caption{Estimated correlation matrix for the CN subjects fitted locally using nonparametric CORE in \eqref{eq:NPM:MODEL} illustrating the dependence of functional connectivity on total cognitive score which gets modulated by age. The arrangement of the panels are the same as that of Figure~\ref{Figure AD_NPM_fit}.}
		\label{A6:Fig:CN_PGM_fit}
	\end{figure}

	\begin{figure}[!htb]
		\centering
		\includegraphics[width=.8\textwidth]{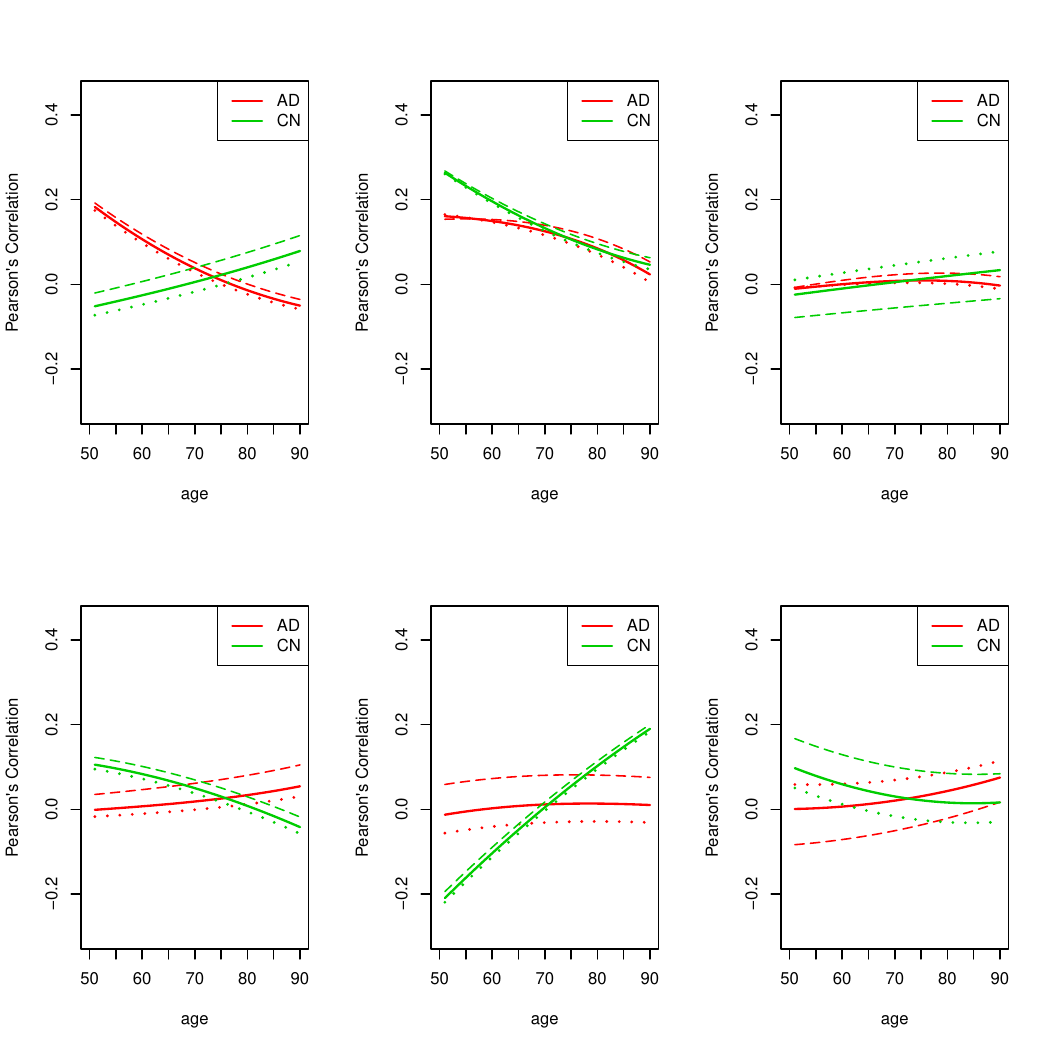}
		\caption{Fitted correlations for varying total cognitive scores and ages across the AD and CN subjects for six chosen connectivity hubs LMF-vs-LPL, RMF-vs-LPL, LPL-vs-RS, PCP-vs-RS, RMT-vs-PCP, MSF-vs-PCP (clockwise in the six panels, starting at upper left). The dashed, solid and dotted lines represent the estimated correlation at $x= 10\%, \ 50\%,\ $ and $90\% \ $ quantiles of the total cognitive score, respectively, for varying ages. For the CN subjects, the inter-hub correlations appear higher for a lower value of the total cognitive score however, for the AD subjects such pattern is not evident. The correlations get generally weaker with higher age.}
		\label{A6:Fig:connectivity}
	\end{figure}
	\clearpage
	The following figure is an additional illustration for the real data application for impact of GDP on human mortality, where a sample of age-at-death densities were treated as the distributional object responses varying with calendar years for $22$ countries and GDP data of each country, for changing calendar year were considered as predictors. The figure shows the 3D plots for the observed age-at-death distributions, represented as densities, over the years for four countries- Australia, Finland, Portugal and the U.S., as is referred to in the main paper in Section~\ref{data:GDP}.
	
	\begin{figure}[H]
		\hspace{-1.1cm}\begin{minipage}{0.5\textwidth}
			\centering
			Australia
			\includegraphics[width=.8\textwidth]{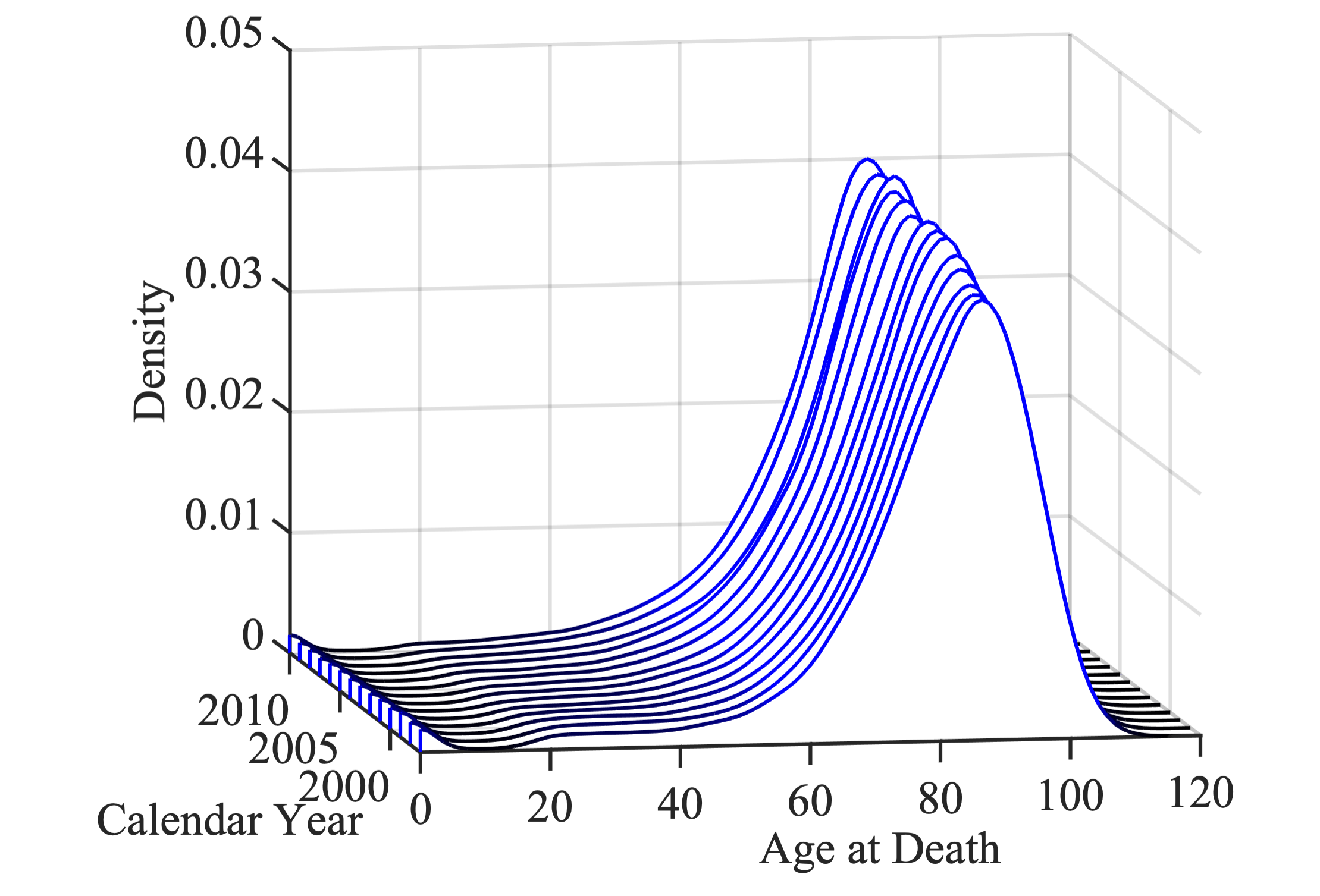}
		\end{minipage}\hfill
		\begin{minipage}{0.5\textwidth}
			\centering
			Finland
			\includegraphics[width=.8\textwidth]{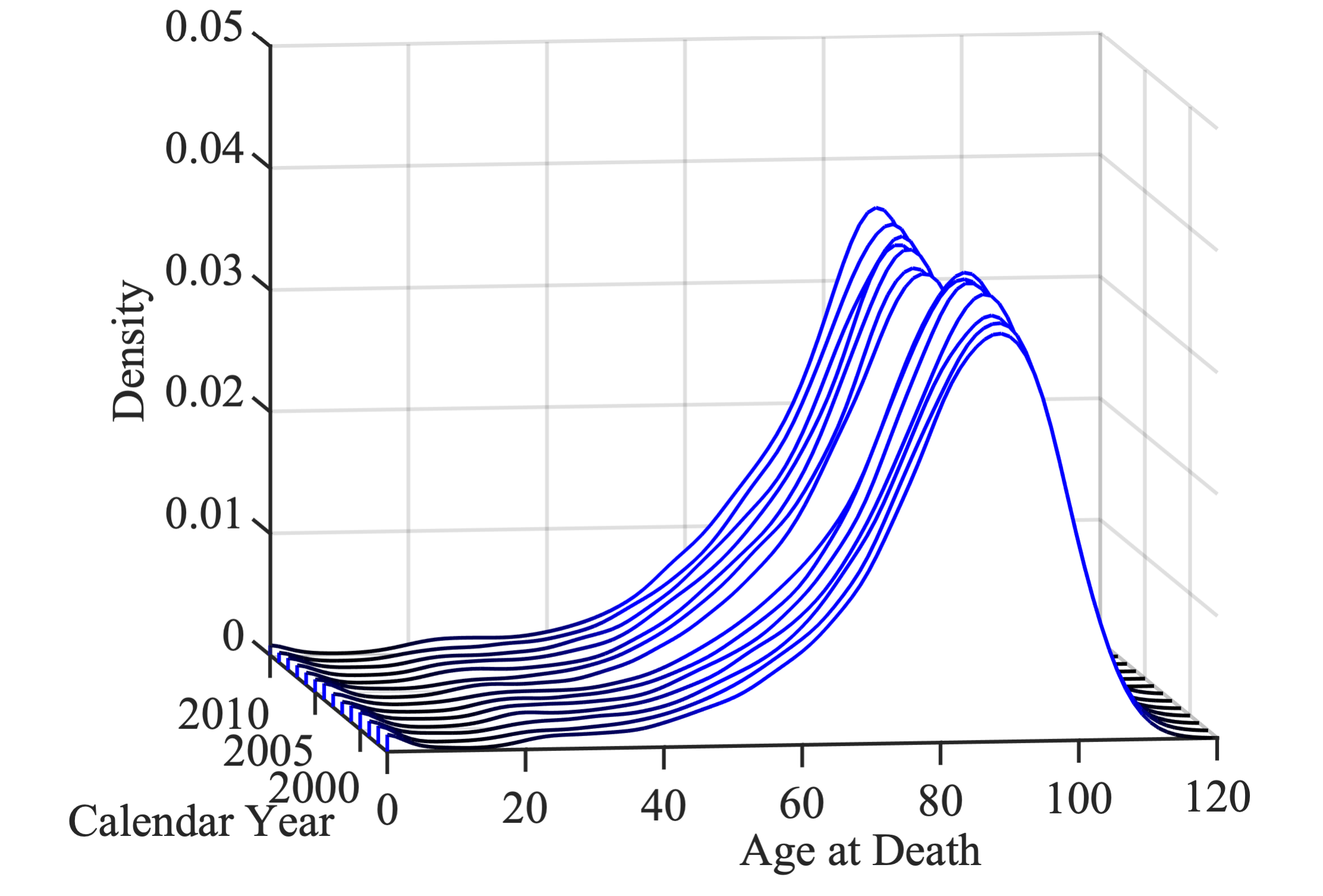}
		\end{minipage}\vspace*{.5cm}
		\hspace*{-1cm}\begin{minipage}{0.5\textwidth}
			\centering
			Portugal
			\includegraphics[width=.8\textwidth]{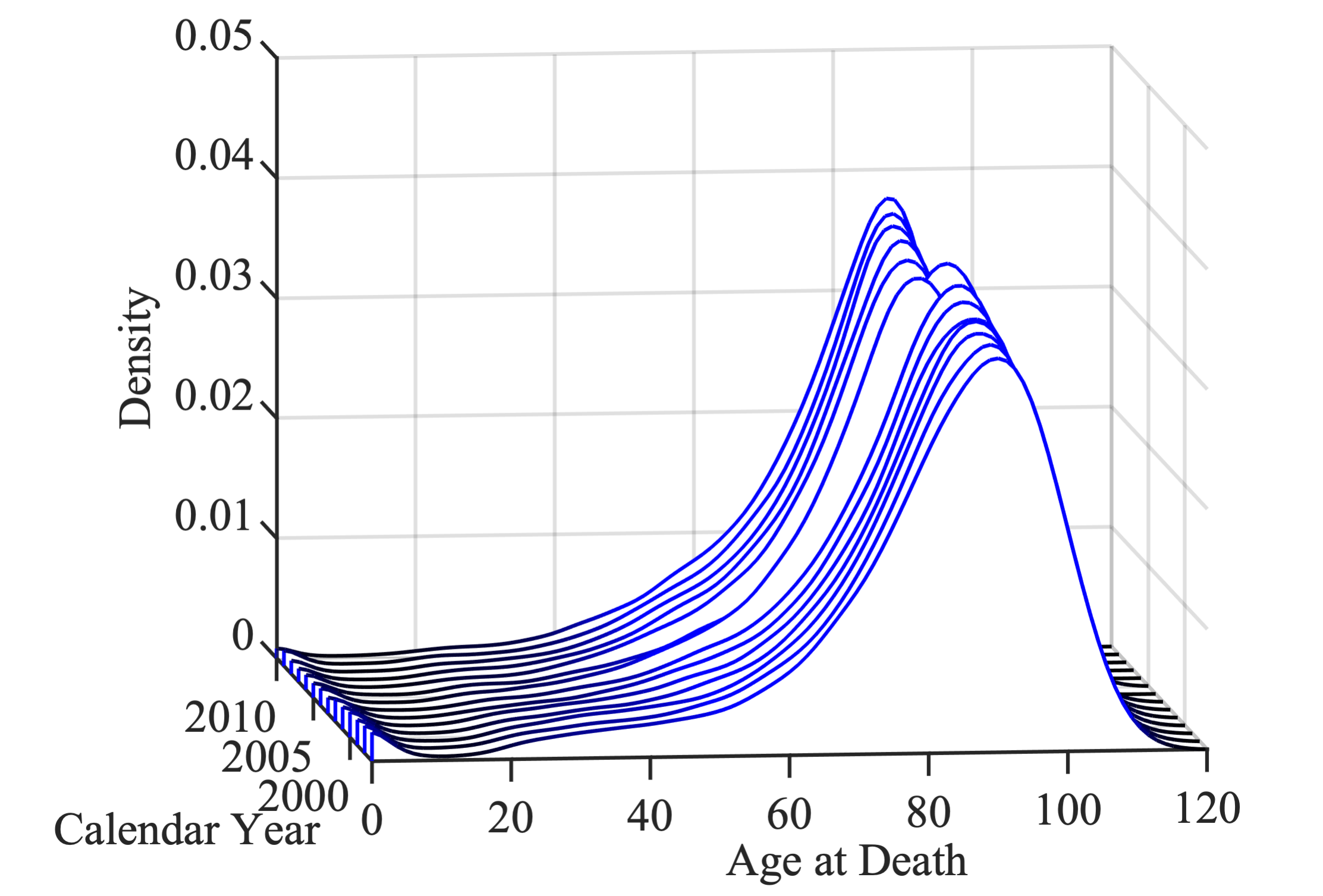}
		\end{minipage}
		\hspace*{1.3cm}\begin{minipage}{0.5\textwidth}
			\centering
			U.S.
			\includegraphics[width=.8\textwidth]{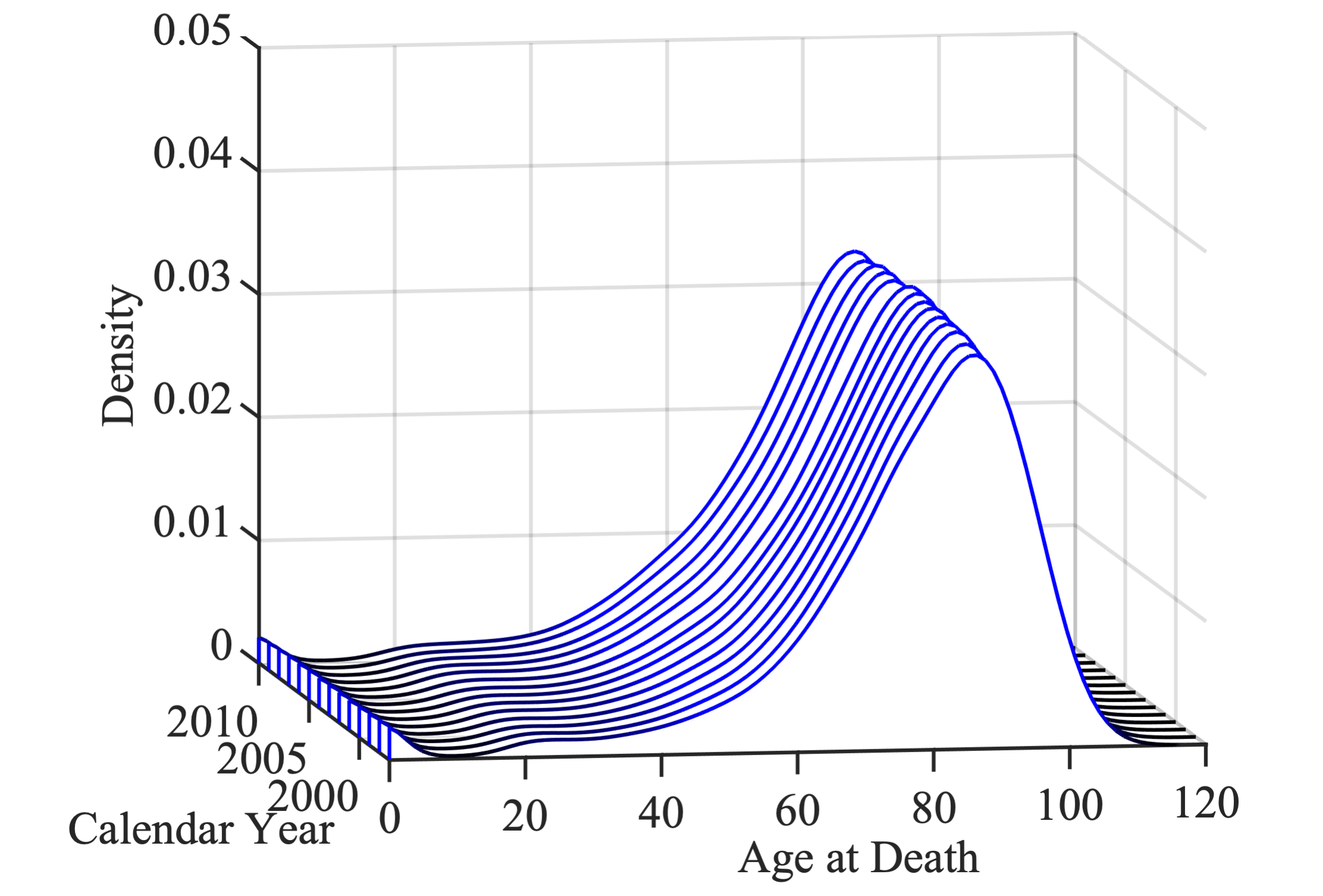}
		\end{minipage}
		\caption{The observed time-varying age at death density functions over the years for males in Australia, Finland, U.S. and Portugal, clockwise in the four panels, starting at the upper left.}
		\label{Figure A1}
	\end{figure}

\bibliography{core}

\begin{thebibliography}{69}
\expandafter\ifx\csname natexlab\endcsname\relax\def\natexlab#1{#1}\fi
\expandafter\ifx\csname url\endcsname\relax
  \def\url#1{\texttt{#1}}\fi
\expandafter\ifx\csname urlprefix\endcsname\relax\def\urlprefix{URL }\fi

\bibitem[{Afsari(2011)}]{afsa:11}
\textsc{Afsari, B.} (2011).
\newblock Riemannian ${L}^p$ center of mass: existence, uniqueness, and
  convexity.
\newblock \textit{Proceedings of the American Mathematical Society}
  \textbf{139} 655--673.

\bibitem[{Allen et~al.(2014)Allen, Damaraju, Plis, Erhardt, Eichele and
  Calhoun}]{allen:14}
\textsc{Allen, E.}, \textsc{Damaraju, E.}, \textsc{Plis, S.}, \textsc{Erhardt,
  E.}, \textsc{Eichele, T.} and \textsc{Calhoun, V.} (2014).
\newblock Tracking whole-brain connectivity dynamics in the resting state.
\newblock \textit{Cerebral Cortex} \textbf{24} 663--676.

\bibitem[{Arsigny et~al.(2007)Arsigny, Fillard, Pennec and Ayache}]{arsi:07}
\textsc{Arsigny, V.}, \textsc{Fillard, P.}, \textsc{Pennec, X.} and
  \textsc{Ayache, N.} (2007).
\newblock Geometric means in a novel vector space structure on symmetric
  positive-definite matrices.
\newblock \textit{SIAM Journal on {M}atrix {A}nalysis and {A}pplications}
  \textbf{29} 328--347.

\bibitem[{Bhattacharya and Patrangenaru(2003)}]{bhat:03}
\textsc{Bhattacharya, R.} and \textsc{Patrangenaru, V.} (2003).
\newblock Large sample theory of intrinsic and extrinsic sample means on
  manifolds.
\newblock \textit{Annals of Statistics} \textbf{31}.

\bibitem[{Bhattacharya and Patrangenaru(2005)}]{bhat:05}
\textsc{Bhattacharya, R.} and \textsc{Patrangenaru, V.} (2005).
\newblock Large sample theory of intrinsic and extrinsic sample means on
  manifolds: {II}.
\newblock \textit{Annals of Statistics} .

\bibitem[{Bickel et~al.(2007)Bickel, Li et~al.}]{bick:07}
\textsc{Bickel, P.~J.}, \textsc{Li, B.} \textsc{et~al.} (2007).
\newblock Local polynomial regression on unknown manifolds.
\newblock In \textit{Complex {D}atasets And {I}nverse {P}roblems}. Institute of
  Mathematical Statistics, 177--186.

\bibitem[{Billera et~al.(2001)Billera, Holmes and Vogtmann}]{bill:01}
\textsc{Billera, L.~J.}, \textsc{Holmes, S.~P.} and \textsc{Vogtmann, K.}
  (2001).
\newblock Geometry of the space of phylogenetic trees.
\newblock \textit{Advances in Applied Mathematics} \textbf{27} 733--767.

\bibitem[{Brier et~al.(2012)Brier, Thomas, Snyder, Benzinger, Zhang, Raichle,
  Holtzman, Morris and Ances}]{brier:12}
\textsc{Brier, M.~R.}, \textsc{Thomas, J.~B.}, \textsc{Snyder, A.~Z.},
  \textsc{Benzinger, T.~L.}, \textsc{Zhang, D.}, \textsc{Raichle, M.~E.},
  \textsc{Holtzman, D.~M.}, \textsc{Morris, J.~C.} and \textsc{Ances, B.~M.}
  (2012).
\newblock Loss of intranetwork and internetwork resting state functional
  connections with {A}lzheimer's disease progression.
\newblock \textit{Journal of {N}euroscience} \textbf{32} 8890--8899.

\bibitem[{Buckner et~al.(2009)Buckner, Sepulcre, Talukdar, Krienen, Liu,
  Hedden, Andrews-Hanna, Sperling and Johnson}]{buck:09}
\textsc{Buckner, R.~L.}, \textsc{Sepulcre, J.}, \textsc{Talukdar, T.},
  \textsc{Krienen, F.~M.}, \textsc{Liu, H.}, \textsc{Hedden, T.},
  \textsc{Andrews-Hanna, J.~R.}, \textsc{Sperling, R.~A.} and \textsc{Johnson,
  K.~A.} (2009).
\newblock Cortical hubs revealed by intrinsic functional connectivity: mapping,
  assessment of stability, and relation to {A}lzheimer's disease.
\newblock \textit{Journal of {N}euroscience} \textbf{29} 1860--1873.

\bibitem[{Chen et~al.(2021)Chen, Lin and M{\"u}ller}]{chen:21}
\textsc{Chen, Y.}, \textsc{Lin, Z.} and \textsc{M{\"u}ller, H.-G.} (2021).
\newblock Wasserstein regression.
\newblock \textit{Journal of the American Statistical Association}  1--14.

\bibitem[{Chiang et~al.(2001)Chiang, Rice and Wu}]{chiang:01}
\textsc{Chiang, C.-T.}, \textsc{Rice, J.~A.} and \textsc{Wu, C.~O.} (2001).
\newblock Smoothing spline estimation for varying coefficient models with
  repeatedly measured dependent variables.
\newblock \textit{Journal of the American Statistical Association} \textbf{96}
  605--619.
\newline\urlprefix\url{http://www.jstor.org/stable/2670300}

\bibitem[{Cleveland et~al.(2017)Cleveland, Grosse and Shyu}]{clev:17}
\textsc{Cleveland, W.~S.}, \textsc{Grosse, E.} and \textsc{Shyu, W.~M.} (2017).
\newblock Local regression models.
\newblock In \textit{Statistical {M}odels In S}. Routledge, 309--376.

\bibitem[{Cornea et~al.(2017)Cornea, Zhu, Kim and Ibrahim}]{cornea:17}
\textsc{Cornea, E.}, \textsc{Zhu, H.}, \textsc{Kim, P.} and \textsc{Ibrahim,
  J.~G.} (2017).
\newblock {Regression models on Riemannian symmetric spaces}.
\newblock \textit{Journal of the Royal Statistical Society Series B}
  \textbf{79} 463--482.
\newline\urlprefix\url{https://ideas.repec.org/a/bla/jorssb/v79y2017i2p463-482.html}

\bibitem[{Damoiseaux et~al.(2012)Damoiseaux, Prater, Miller and
  Greicius}]{damo:12}
\textsc{Damoiseaux, J.~S.}, \textsc{Prater, K.~E.}, \textsc{Miller, B.~L.} and
  \textsc{Greicius, M.~D.} (2012).
\newblock Functional connectivity tracks clinical deterioration in
  {A}lzheimer's disease.
\newblock \textit{Neurobiology of {A}ging} \textbf{33} 828--e19.

\bibitem[{Davis et~al.(2005)Davis, Foskey, Rosenman, Goyal, Chang and
  Joshi}]{davi:05}
\textsc{Davis, B.~C.}, \textsc{Foskey, M.}, \textsc{Rosenman, J.},
  \textsc{Goyal, L.}, \textsc{Chang, S.} and \textsc{Joshi, S.} (2005).
\newblock Automatic segmentation of intra-treatment {CT} images for adaptive
  radiation therapy of the prostate.
\newblock In \textit{International Conference on Medical Image Computing and
  Computer-Assisted Intervention}. Springer.

\bibitem[{Di~Marzio et~al.(2014)Di~Marzio, Panzera and Taylor}]{di:14}
\textsc{Di~Marzio, M.}, \textsc{Panzera, A.} and \textsc{Taylor, C.~C.} (2014).
\newblock Nonparametric regression for spherical data.
\newblock \textit{Journal of the American Statistical Association} \textbf{109}
  748--763.

\bibitem[{Dryden et~al.(2009)Dryden, Koloydenko, Zhou et~al.}]{dryd:09}
\textsc{Dryden, I.~L.}, \textsc{Koloydenko, A.}, \textsc{Zhou, D.}
  \textsc{et~al.} (2009).
\newblock Non-{E}uclidean statistics for covariance matrices, with applications
  to diffusion tensor imaging.
\newblock \textit{Annals of Applied Statistics} \textbf{3} 1102--1123.

\bibitem[{Dryden et~al.(2010)Dryden, Kolydenko, Zhou and Li}]{dryd:10}
\textsc{Dryden, I.~L.}, \textsc{Kolydenko, A.}, \textsc{Zhou, D.} and
  \textsc{Li, B.} (2010).
\newblock Non-{E}uclidean statistical analysis of covariance matrices and
  diffusion tensors.
\newblock \textit{arXiv preprint arXiv:1010.3955} .

\bibitem[{Eubank et~al.(2004)Eubank, Huang, Maldonado, Wang, Wang and
  Buchanan}]{eubank:04}
\textsc{Eubank, R.}, \textsc{Huang, C.}, \textsc{Maldonado, Y.~M.},
  \textsc{Wang, N.}, \textsc{Wang, S.} and \textsc{Buchanan, R.} (2004).
\newblock Smoothing spline estimation in varying-coefficient models.
\newblock \textit{Journal of the Royal Statistical Society: Series B
  (Statistical Methodology)} \textbf{66} 653--667.

\bibitem[{Fan and Gijbels(1996)}]{fan:96}
\textsc{Fan, J.} and \textsc{Gijbels, I.} (1996).
\newblock \textit{Local Polynomial Modelling and its Applications: Monographs
  on Statistics and Applied Arobability 66}.
\newblock Routledge.

\bibitem[{Fan and Zhang(2000)}]{fan:00}
\textsc{Fan, J.} and \textsc{Zhang, J.-T.} (2000).
\newblock Two-step estimation of functional linear models with applications to
  longitudinal data.
\newblock \textit{Journal of the Royal Statistical Society: Series B
  (Statistical Methodology)} \textbf{62} 303--322.

\bibitem[{Fan and Zhang(1999)}]{fan:99}
\textsc{Fan, J.} and \textsc{Zhang, W.} (1999).
\newblock Statistical estimation in varying coefficient models.
\newblock \textit{Annals of Statistics} \textbf{27} 1491--1518.
\newline\urlprefix\url{https://doi.org/10.1214/aos/1017939139}

\bibitem[{Fan and Zhang(2008)}]{fan:zhang:08}
\textsc{Fan, J.} and \textsc{Zhang, W.} (2008).
\newblock Statistical methods with varying coefficient models.
\newblock \textit{Statistics and its Interface} \textbf{1} 179.

\bibitem[{Ferreira and Busatto(2013)}]{ferr:busa:13}
\textsc{Ferreira, L. R.~K.} and \textsc{Busatto, G.~F.} (2013).
\newblock Resting-state functional connectivity in normal brain aging.
\newblock \textit{{N}euroscience \& {B}iobehavioral {R}eviews} \textbf{37}
  384--400.

\bibitem[{Fr\'echet(1948)}]{frechet:48}
\textsc{Fr\'echet, M.~R.} (1948).
\newblock Les \'el\'ements al\'eatoires de nature quelconque dans un espace
  distanci\'e.
\newblock \textit{Annales de l'institut Henri Poincar\'e} \textbf{10} 215--310.
\newline\urlprefix\url{http://www.numdam.org/item/AIHP_1948__10_4_215_0}

\bibitem[{Hastie and Tibshirani(1993)}]{hast:tibs:93}
\textsc{Hastie, T.} and \textsc{Tibshirani, R.} (1993).
\newblock Varying-coefficient models.
\newblock \textit{Journal of the Royal Statistical Society. Series B
  (Methodological)} \textbf{55} 757--796.
\newline\urlprefix\url{http://www.jstor.org/stable/2345993}

\bibitem[{Hoover et~al.(1998)Hoover, Rice, Wu and Yang}]{hoov:98}
\textsc{Hoover, D.~R.}, \textsc{Rice, J.~A.}, \textsc{Wu, C.~O.} and
  \textsc{Yang, L.-P.} (1998).
\newblock {Nonparametric smoothing estimates of time-varying coefficient models
  with longitudinal data}.
\newblock \textit{Biometrika} \textbf{85} 809--822.
\newline\urlprefix\url{https://doi.org/10.1093/biomet/85.4.809}

\bibitem[{Horv{\'a}th and Kokoszka(2012)}]{horv:koko:12}
\textsc{Horv{\'a}th, L.} and \textsc{Kokoszka, P.} (2012).
\newblock \textit{Inference For Functional Data With Applications}.
\newblock Springer Science \& Business Media.

\bibitem[{Huang et~al.(2002)Huang, Wu and Zhou}]{huang:02}
\textsc{Huang, J.~Z.}, \textsc{Wu, C.~O.} and \textsc{Zhou, L.} (2002).
\newblock Varying-coefficient models and basis function approximations for the
  analysis of repeated measurements.
\newblock \textit{Biometrika} \textbf{89} 111--128.
\newline\urlprefix\url{http://www.jstor.org/stable/4140562}

\bibitem[{Kloeckner(2010)}]{kloe:10}
\textsc{Kloeckner, B.} (2010).
\newblock A geometric study of {W}asserstein spaces: Euclidean spaces.
\newblock \textit{Annali della Scuola Normale Superiore di Pisa-Classe di
  Scienze} \textbf{9} 297--323.

\bibitem[{Kueper et~al.(2018)Kueper, Speechley and Montero-Odasso}]{kuep:18}
\textsc{Kueper, J.~K.}, \textsc{Speechley, M.} and \textsc{Montero-Odasso, M.}
  (2018).
\newblock The {A}lzheimer’s disease assessment scale--cognitive subscale
  (adas-cog): modifications and responsiveness in pre-dementia populations. a
  narrative review.
\newblock \textit{Journal of {A}lzheimer's Disease} \textbf{63} 423--444.

\bibitem[{Lin(2019)}]{lin:19}
\textsc{Lin, Z.} (2019).
\newblock Riemannian {g}eometry of {s}ymmetric {p}ositive {d}efinite {m}atrices
  via {C}holesky {d}ecomposition.
\newblock \textit{SIAM Journal on Matrix Analysis and Applications} \textbf{40}
  1353--1370.

\bibitem[{Maity(2017)}]{maity:17}
\textsc{Maity, A.} (2017).
\newblock Nonparametric functional concurrent regression models.
\newblock \textit{Wiley Interdisciplinary Reviews: Computational Statistics}
  \textbf{9} e1394.

\bibitem[{Manrique et~al.(2018)Manrique, Crambes and Hilgert}]{manr:18}
\textsc{Manrique, T.}, \textsc{Crambes, C.} and \textsc{Hilgert, N.} (2018).
\newblock Ridge regression for the functional concurrent model.
\newblock \textit{Electronic Journal of Statistics} \textbf{12} 985--1018.

\bibitem[{Marron and Alonso(2014)}]{marr:alon:14}
\textsc{Marron, J.~S.} and \textsc{Alonso, A.~M.} (2014).
\newblock Overview of object oriented data analysis.
\newblock \textit{Biometrical Journal} \textbf{56} 732--753.
\newline\urlprefix\url{https://onlinelibrary.wiley.com/doi/abs/10.1002/bimj.201300072}

\bibitem[{Moakher(2005)}]{moak:05}
\textsc{Moakher, M.} (2005).
\newblock A differential geometric approach to the geometric mean of symmetric
  positive-definite matrices.
\newblock \textit{SIAM Journal on Matrix Analysis and Applications} \textbf{26}
  735--747.

\bibitem[{M\"uller(2016)}]{mull:16}
\textsc{M\"uller, H.-G.} (2016).
\newblock Peter {H}all, functional data analysis and random objects.
\newblock \textit{Annals of Statistics} \textbf{44} 1867--1887.
\newline\urlprefix\url{https://doi.org/10.1214/16-AOS1492}

\bibitem[{Nietert et~al.(2021)Nietert, Goldfeld and Kato}]{kato:21}
\textsc{Nietert, S.}, \textsc{Goldfeld, Z.} and \textsc{Kato, K.} (2021).
\newblock Smooth p-wasserstein distance: structure, empirical approximation,
  and statistical applications.
\newblock In \textit{International Conference on Machine Learning}. PMLR.

\bibitem[{Patrangenaru and Ellingson(2015)}]{patr:15}
\textsc{Patrangenaru, V.} and \textsc{Ellingson, L.} (2015).
\newblock \textit{Nonparametric Statistics On Manifolds And Their Applications
  To Object Data Analysis}.
\newblock CRC Press.

\bibitem[{Pennec(2018)}]{penn:18}
\textsc{Pennec, X.} (2018).
\newblock Barycentric subspace analysis on manifolds.
\newblock \textit{Annals of Statistics} \textbf{46} 2711--2746.

\bibitem[{Pennec et~al.(2006)Pennec, Fillard and Ayache}]{penn:06}
\textsc{Pennec, X.}, \textsc{Fillard, P.} and \textsc{Ayache, N.} (2006).
\newblock A {R}iemannian framework for tensor computing.
\newblock \textit{International {J}ournal of {C}omputer {V}ision} \textbf{66}
  41--66.

\bibitem[{Petersen and M\"uller(2016)}]{pete:mull:16}
\textsc{Petersen, A.} and \textsc{M\"uller, H.-G.} (2016).
\newblock Functional data analysis for density functions by transformation to a
  hilbert space.
\newblock \textit{Annals of Statistics} \textbf{44} 183--218.
\newline\urlprefix\url{https://doi.org/10.1214/15-AOS1363}

\bibitem[{Petersen and M\"uller(2019)}]{pete:mull:19}
\textsc{Petersen, A.} and \textsc{M\"uller, H.-G.} (2019).
\newblock Fr\'echet regression for random objects with {E}uclidean predictors.
\newblock \textit{Annals of Statistics} \textbf{47} 691--719.
\newline\urlprefix\url{https://doi.org/10.1214/17-AOS1624}

\bibitem[{Pigoli et~al.(2014)Pigoli, Aston, Dryden and Secchi}]{pigo:14}
\textsc{Pigoli, D.}, \textsc{Aston, J.~A.}, \textsc{Dryden, I.~L.} and
  \textsc{Secchi, P.} (2014).
\newblock Distances and inference for covariance operators.
\newblock \textit{Biometrika} \textbf{101} 409--422.

\bibitem[{Ramsay and Silverman(2005)}]{rams:05}
\textsc{Ramsay, J.~O.} and \textsc{Silverman, B.~W.} (2005).
\newblock \textit{{Functional Data Analysis}}.
\newblock 2nd ed. Springer Series in Statistics, Springer.
\newline\urlprefix\url{http://www.amazon.com/exec/obidos/redirect?tag=citeulike07-20\&path=ASIN/038740080X}

\bibitem[{Ramsay and Silverman(2007)}]{rams:07}
\textsc{Ramsay, J.~O.} and \textsc{Silverman, B.~W.} (2007).
\newblock \textit{{A}pplied {F}unctional {D}ata {A}nalysis: {M}ethods and
  {C}ase {S}tudies}.
\newblock Springer.

\bibitem[{Scarapicchia et~al.(2018)Scarapicchia, Mazerolle, Fisk, Ritchie and
  Gawryluk}]{scara:18}
\textsc{Scarapicchia, V.}, \textsc{Mazerolle, E.~L.}, \textsc{Fisk, J.~D.},
  \textsc{Ritchie, L.~J.} and \textsc{Gawryluk, J.~R.} (2018).
\newblock Resting state bold variability in {A}lzheimer’s disease: a marker
  of cognitive decline or cerebrovascular status?
\newblock \textit{Frontiers in {A}ging {N}euroscience} \textbf{10} 39.

\bibitem[{Sch{\"o}tz(2020)}]{scho:20}
\textsc{Sch{\"o}tz, C.} (2020).
\newblock Strong laws of large numbers for generalizations of fr\'echet mean
  sets.
\newblock \textit{arXiv preprint arXiv:2012.12762} .

\bibitem[{Sent\"urk and M\"uller(2010)}]{sent:mllr:10}
\textsc{Sent\"urk, D.} and \textsc{M\"uller, H.-G.} (2010).
\newblock Functional varying coefficient models for longitudinal data.
\newblock \textit{Journal of the American Statistical Association} \textbf{105}
  1256--1264.
\newline\urlprefix\url{https://doi.org/10.1198/jasa.2010.tm09228}

\bibitem[{{S}ent\"urk and Nguyen(2011)}]{sent:nguy:11}
\textsc{{S}ent\"urk, D.} and \textsc{Nguyen, D.~V.} (2011).
\newblock Varying coefficient models for sparse noise-contaminated longitudinal
  data.
\newblock \textit{Statistica Sinica} \textbf{21} 1831--1856.
\newline\urlprefix\url{http://www.jstor.org/stable/24309657}

\bibitem[{Shi et~al.(2007)Shi, Wang, Murray-Smith and Titterington}]{shi:07}
\textsc{Shi, J.~Q.}, \textsc{Wang, B.}, \textsc{Murray-Smith, R.} and
  \textsc{Titterington, D.~M.} (2007).
\newblock Gaussian process functional regression modeling for batch data.
\newblock \textit{Biometrics} \textbf{63} 714--723.
\newline\urlprefix\url{http://www.jstor.org/stable/4541403}

\bibitem[{Sturm(2003)}]{stur:03}
\textsc{Sturm, K.-T.} (2003).
\newblock Probability measures on metric spaces of nonpositive.
\newblock \textit{Heat Kernels and Analysis on Manifolds, Graphs, and Metric
  Spaces: Lecture Notes from a Quarter Program on Heat Kernels, Random Walks,
  and Analysis on Manifolds and Graphs: April 16-July 13, 2002, Emile Borel
  Centre of the Henri Poincar{\'e} Institute, Paris, France} \textbf{338} 357.

\bibitem[{Turner et~al.(2014)Turner, Mileyko, Mukherjee and Harer}]{sayan:14}
\textsc{Turner, K.}, \textsc{Mileyko, Y.}, \textsc{Mukherjee, S.} and
  \textsc{Harer, J.} (2014).
\newblock Fr\'echet means for distributions of persistence diagrams.
\newblock \textit{{Discrete \& Computational Geometry}} \textbf{52} 44--70.

\bibitem[{Verzelen et~al.(2012)Verzelen, Tao and M\"uller}]{verz:12}
\textsc{Verzelen, N.}, \textsc{Tao, W.} and \textsc{M\"uller, H.-G.} (2012).
\newblock Inferring stochastic dynamics from functional data.
\newblock \textit{Biometrika} \textbf{99} 533--550.
\newline\urlprefix\url{http://www.jstor.org/stable/41720713}

\bibitem[{Wang and Shi(2014)}]{wang:shi:14}
\textsc{Wang, B.} and \textsc{Shi, J.~Q.} (2014).
\newblock Generalized {G}aussian process regression model for non-{G}aussian
  functional data.
\newblock \textit{Journal of the American Statistical Association} \textbf{109}
  1123--1133.

\bibitem[{Wang et~al.(2016)Wang, Chiou and M\"uller}]{wang:16}
\textsc{Wang, J.-L.}, \textsc{Chiou, J.-M.} and \textsc{M\"uller, H.-G.}
  (2016).
\newblock Functional data analysis.
\newblock \textit{Annual Review of Statistics and its Application} \textbf{3}
  257--295.

\bibitem[{Wang et~al.(2007)Wang, Liang, Wang, Tian, Zhang, Li and
  Jiang}]{wang:07}
\textsc{Wang, K.}, \textsc{Liang, M.}, \textsc{Wang, L.}, \textsc{Tian, L.},
  \textsc{Zhang, X.}, \textsc{Li, K.} and \textsc{Jiang, T.} (2007).
\newblock Altered functional connectivity in early {A}lzheimer's disease: A
  resting-state f{MRI} study.
\newblock \textit{Human {B}rain {M}apping} \textbf{28} 967--978.

\bibitem[{Wang et~al.(2008)Wang, Li and Huang}]{wang:08}
\textsc{Wang, L.}, \textsc{Li, H.} and \textsc{Huang, J.~Z.} (2008).
\newblock Variable selection in nonparametric varying-coefficient models for
  analysis of repeated measurements.
\newblock \textit{Journal of the American Statistical Association} \textbf{103}
  1556--1569.

\bibitem[{Wu and Chiang(2000)}]{wu:chiang:00}
\textsc{Wu, C.~O.} and \textsc{Chiang, C.-T.} (2000).
\newblock Kernel smoothing on varying coefficient models with longitudinal
  dependent variable.
\newblock \textit{Statistica Sinica}  433--456.

\bibitem[{Yang and Vemuri(2020)}]{yang:20}
\textsc{Yang, C.-H.} and \textsc{Vemuri, B.~C.} (2020).
\newblock Shrinkage estimation of the fr\'echet mean in lie groups.
\newblock \textit{arXiv preprint arXiv:2009.13020} .

\bibitem[{Yao et~al.(2005)Yao, M{\"u}ller and Wang}]{yao:05}
\textsc{Yao, F.}, \textsc{M{\"u}ller, H.-G.} and \textsc{Wang, J.-L.} (2005).
\newblock Functional data analysis for sparse longitudinal data.
\newblock \textit{Journal of the American Statistical Association} \textbf{100}
  577--590.

\bibitem[{Yuan et~al.(2013)Yuan, Zhu, Styner, Gilmore and Marron}]{yuan:13}
\textsc{Yuan, Y.}, \textsc{Zhu, H.}, \textsc{Styner, M.}, \textsc{Gilmore,
  J.~H.} and \textsc{Marron, J.~S.} (2013).
\newblock Varying coefficient model for modeling diffusion tensors along white
  matter tracts.
\newblock \textit{Annals of Applied Statistics} \textbf{7} 102--125.
\newline\urlprefix\url{https://doi.org/10.1214/12-AOAS574}

\bibitem[{Zhang et~al.(2010)Zhang, Wang, Liu, Ma, Yang, Zhang and
  Teng}]{zhang:10}
\textsc{Zhang, H.-Y.}, \textsc{Wang, S.-J.}, \textsc{Liu, B.}, \textsc{Ma,
  Z.-L.}, \textsc{Yang, M.}, \textsc{Zhang, Z.-J.} and \textsc{Teng, G.-J.}
  (2010).
\newblock Resting brain connectivity: changes during the progress of
  {A}lzheimer disease.
\newblock \textit{Radiology} \textbf{256} 598--606.

\bibitem[{Zhang et~al.(2011)Zhang, Clayton and Townsend}]{zhang:11}
\textsc{Zhang, J.}, \textsc{Clayton, M.~K.} and \textsc{Townsend, P.~A.}
  (2011).
\newblock Functional concurrent linear regression model for spatial images.
\newblock \textit{Journal of Agricultural, Biological, and Environmental
  Statistics} \textbf{16} 105--130.

\bibitem[{Zhang et~al.(2021)Zhang, Xue and Li}]{zhan:21}
\textsc{Zhang, Q.}, \textsc{Xue, L.} and \textsc{Li, B.} (2021).
\newblock Dimension reduction and data visualization for fr\'echet regression.
\newblock \textit{arXiv preprint arXiv:2110.00467} .

\bibitem[{Zhou et~al.(2016)Zhou, Dryden, Koloydenko, Audenaert and
  Bai}]{zhou:16}
\textsc{Zhou, D.}, \textsc{Dryden, I.~L.}, \textsc{Koloydenko, A.~A.},
  \textsc{Audenaert, K.~M.} and \textsc{Bai, L.} (2016).
\newblock Regularisation, interpolation and visualisation of diffusion tensor
  images using non-{E}uclidean statistics.
\newblock \textit{Journal of Applied Statistics} \textbf{43} 943--978.

\bibitem[{Zhou et~al.(2010)Zhou, Greicius, Gennatas, Growdon, Jang, Rabinovici,
  Kramer, Weiner, Miller and Seeley}]{zhou:10}
\textsc{Zhou, J.}, \textsc{Greicius, M.~D.}, \textsc{Gennatas, E.~D.},
  \textsc{Growdon, M.~E.}, \textsc{Jang, J.~Y.}, \textsc{Rabinovici, G.~D.},
  \textsc{Kramer, J.~H.}, \textsc{Weiner, M.}, \textsc{Miller, B.~L.} and
  \textsc{Seeley, W.~W.} (2010).
\newblock Divergent network connectivity changes in behavioural variant
  frontotemporal dementia and {Alzheimer’s} disease.
\newblock \textit{Brain} \textbf{133} 1352--1367.

\bibitem[{Zhu et~al.(2009)Zhu, Chen, Ibrahim, Li, Hall and Lin}]{zhu:09}
\textsc{Zhu, H.}, \textsc{Chen, Y.}, \textsc{Ibrahim, J.~G.}, \textsc{Li, Y.},
  \textsc{Hall, C.} and \textsc{Lin, W.} (2009).
\newblock Intrinsic regression models for positive-definite matrices with
  applications to diffusion tensor imaging.
\newblock \textit{Journal of the American Statistical Association} \textbf{104}
  1203--1212.
\newline\urlprefix\url{https://doi.org/10.1198/jasa.2009.tm08096}

\bibitem[{Zhu et~al.(2010)Zhu, Styner, Li, Kong, Shi, Lin, Coe and
  Gilmore}]{zhu:10}
\textsc{Zhu, H.}, \textsc{Styner, M.}, \textsc{Li, Y.}, \textsc{Kong, L.},
  \textsc{Shi, Y.}, \textsc{Lin, W.}, \textsc{Coe, C.} and \textsc{Gilmore,
  J.~H.} (2010).
\newblock Multivariate varying coefficient models for {DTI} tract statistics.
\newblock In \textit{International Conference on Medical Image Computing and
  Computer-Assisted Intervention}. Springer.

\end{thebibliography}
\end{document}